\documentclass[
]{ceurart}

\newif\ifextended
\extendedtrue

\newif\ifanonymous
\anonymousfalse

\usepackage[english]{babel}

\ifextended
  \usepackage[bibliography=common]{apxproof}
\else
  \usepackage[bibliography=common,appendix=strip]{apxproof}
\fi

\usepackage{url}
\usepackage[utf8]{inputenc} %
\usepackage[small]{caption} %
\usepackage{graphicx}
\usepackage{amsmath, amssymb, amsfonts, amsthm}
\usepackage{mathtools}
\usepackage{booktabs}
\usepackage{algorithm}
\usepackage{algorithmic}
\usepackage{xspace}
\usepackage{wrapfig}

\usepackage{subcaption}

\usepackage{natbib}
\usepackage{amsfonts}

\usepackage[textwidth=13mm,tickmarkheight=1em,disable]{todonotes}

\usepackage[capitalise]{cleveref}

\usepackage{enumitem}
\usepackage{tikz}
\usetikzlibrary{arrows,automata,positioning}
\usepackage{accents}

\newtheoremrep{theorem}{Theorem}
\newtheorem{example}[theorem]{Example}
\newtheoremrep{lemma}[theorem]{Lemma}
\newtheoremrep{proposition}[theorem]{Proposition}
\newtheoremrep{definition}[theorem]{Definition}
\newtheoremrep{observation}[theorem]{Observation}
\newtheoremrep{corollary}[theorem]{Corollary}

\title{Constructive Preference Relations: Navigating Undecidability in Rational LTL Contraction}

\ifanonymous
\author{(Anonymous Authors)}
\else
\author[1]{Hannes Gaißer}[%
email=hannes-carl.gaisser@student.uni-tuebingen.de,
]
\address[1]{University of T\"ubingen, Germany}

\author[2]{Dominik Klumpp}[%
email=klumpp@lix.polytechnique.fr,
]
\address[2]{LIX -- CNRS -- \'Ecole Polytechnique, France}

\author[3]{Jandson S Ribeiro}[%
email=ribeiroj@cardiff.ac.uk,
]
\address[3]{Cardiff University, United Kingdom}
\fi

\newcommand\myskip[2]{}

\newcommand\quot[1]{\emph{``#1''}}

\newcommand\yes{\texttt{\upshape yes}\xspace}
\newcommand\no{\texttt{\upshape no}\xspace}

\newcommand\decproblem[2]{%
\begin{center}%
  \fbox{\begin{minipage}{0.95\columnwidth}%
    \itshape
    \begin{description}%
      \item[\itshape Input:] #1.\vspace*{-0.75em}%
      \item[\itshape Output:] \yes if #2, \no otherwise.%
    \end{description}%
  \end{minipage}}%
\end{center}%
}%

\newcommand\incomp[1][]{\mathrel{\#_{#1}}}

\makeatletter
\newcommand*{\propdef}[2][\@nil]{%
  \def\tmp{#1}%
  \ifx\tmp\@nnil%
    \expandafter\gdef\csname propref#2\endcsname{\hyperref[prop:#2]{\textbf{\upshape(#2)}}}%
    \textbf{\upshape(#2)}%
  \else%
    \expandafter\gdef\csname propref#2\endcsname{\hyperref[prop:#2]{\textbf{\upshape(#1)}}}%
    \textbf{\upshape(#1)}%
  \fi%
  \phantomsection\label{prop:#2}%
}
\makeatother

\newcommand\prop[1]{%
  \ifcsname propref#1\endcsname%
    \csname propref#1\endcsname%
  \else%
    \errmessage{Property #1 has not been defined (use \textbackslash{}propdef to define it).}%
  \fi%
}

\newcommand\powerset[1]{\mathcal{P}(#1)}
\newcommand\N{\mathbb{N}}
\newcommand\Z{\mathbb{Z}}
\newcommand\R{\mathbb{R}}

\newcommand\Fm{\mathit{Fm}}
\newcommand\CnFun{\mathit{Cn}}
\newcommand\Cn[1]{\CnFun(#1)}
\newcommand\CCT[1][]{\mathit{C\!C\!T}_{\!#1}}
\newcommand\kb{\mathcal{K}}

\newcommand\limplies\to

\newcommand\compl[1]{\overline{\omega}(#1)}

\newcommand\exc{\mathbb{E}}

\newcommand\AP{\mathit{AP}}
\newcommand\LTL{\mathit{L\!T\!L}}

\newcommand\FmLTL{\mathit{Fm}_{\LTL}}
\newcommand\CnLTL[1]{\CnFun_{\LTL}(#1)}

\newcommand{\Next}[1][]{\mathbf{X}^{#1}\,}
\newcommand\Globally{\mathbf{G}\,}
\newcommand\Finally{\mathbf{F}\,}
\newcommand\Until{\mathbin{\mathbf{U}}}

\newcommand\id[1]{\mathit{id}(#1)}

\newcommand\UP{\mathit{UP}}
\newcommand{\lang}[1]{\mathcal{L}(#1)}

\newcommand\init{\mathsf{init}}

\newcommand\Traces[1]{\mathit{Traces}(#1)}
\newcommand\ltlbuchi[1]{A_{#1}}
\newcommand\kripkebuchi[1]{A_{#1}}

\newcommand\support[1]{\mathcal{S}(#1)}
\newcommand\buchi{\textnormal{B\"uchi}}
\newcommand\excBuchi{\exc_\buchi}

\newcommand{\relFn}{\mathcal{R}}
\newcommand{\rel}[1]{{\relFn(#1)}}

\newcommand{\cp}[1][1]{$\mathbf{(K^{-}_{#1})}$}

\makeatletter
\newcommand{\oset}[3][0ex]{%
  \mathbin{\mathop{#3}\limits^{
    \vbox to#1{\kern-2\ex@
    \hbox{$\scriptstyle#2$}\vss}}}}
\makeatother
\newcommand\dotmin[1][]{\oset[-.7ex]{\textbf{\large.}}{-}_{#1}}

\newcommand\pref{\prec}
\newcommand\npref{\not\pref}

\newcommand{\leftend}{\mathord{<}}
\newcommand{\rightend}{\mathord{>}}
\newcommand{\Iota}{\textnormal{I}}
\newcommand\headp[1]{\accentset{\bullet}{#1}}

\DeclareMathOperator{\enc}{enc}

\newcommand\earliestFuture[2]{\mathsf{fst}_{#1}(#2)}

\newcommand{\tracesof}[1]{\lVert  #1 \rVert}
\newcommand{\trace}{\tau}

\begin{document}

\copyrightyear{2026}
\copyrightclause{Copyright for this paper by its authors.
  Use permitted under Creative Commons License Attribution 4.0
  International (CC BY 4.0).}

\conference{24th International
	Workshop on Nonmonotonic Reasoning, July 17–19, 2026, Lisbon,
	Portugal}

\maketitle

\pagestyle{plain}
\thispagestyle{plain}

\begin{abstract}
We study the computational aspects of \emph{epistemic preference relations} in non-classical logics,
particularly \emph{linear temporal logic} (LTL).
Epistemic preferences form the backbone of \emph{belief contraction operators},
which describe how to rationally relinquish obsolete beliefs.
These preference relations have to satisfy certain innocuous conditions;
and constructing such relations is usually assumed to be a trivial process.
However, in the case of LTL, where relations are represented with Büchi automata, we show that this is a challenging task:
the core condition, which guarantees the success of contraction, is in fact \emph{undecidable}.
Towards achieving effective LTL belief contraction,
we then propose several concrete constructions of novel preference relations
that satisfy the required conditions \emph{by design}.
These constructions include, among others,
(1) generalisations of distance measures (e.g.\ Dalal) beyond the classical setting,
as well as (2) the ability to hierarchically compose different preference relations.
Our results not only provide rich families of preference relations for LTL, but also generalise the limited pool of concrete preference relations for the classical cases, allowing us to go beyond Dalal to achieve full rationality.
\end{abstract}

\section{Introduction}
In \emph{belief change} \citep{gardenfors:flux,hansson:belief-dynamics}, the process of relinquishing an obsolete belief is studied under the name of \emph{contraction}.
Ideally, only beliefs associated with the entailment of the obsolete belief should be discarded, so the incurred changes are minimised.
This principle of minimal change is conceptualised via sets of rationality postulates
that prescribe what a rational change is.
Families of constructive operators that abide by such postulates were identified, called rational operators.
Contraction is central in belief change, as several other change processes internalise contraction. For example,  \emph{revision}, the process of accommodating a new belief while maintaining consistency, consists of contracting  potential conflicts with the new belief and then incorporating it.
This makes contraction  an ideal component %
for identifying and understanding how modifications are conducted and, therefore, to properly address minimal change.
We consider belief contraction in the setting of reasoning with temporal information, specifically \emph{linear temporal logic} (LTL, for short)~ \citep{pnueli:ltl}. 

Preference relations are the \emph{de facto} strategy to construct rational change operators. %
For instance, in \emph{iterated belief revision} \citep{DarwicheP97,SchwindKP22} whose focus is on the impact of the order in which new beliefs arrive, %
total relations and rankings%
\footnote{A ranking is a special kind of preference relation where each object is assigned a confidence/preference number.}
are the canonical way %
for representing an agent's beliefs,
while the constructive operators essentially
bend a preference relation into a new one. %
Although most of the research in belief change concentrates in the classical settings \citep{DelgrandePW18, BoothMVW14}, such as classical propositional logic, first order logic and some sub-classical cases, such  as Horn, recently efforts have been put to extend belief change to more expressive logics capable of expressing intricate knowledge domains  such as description logics \citep{DLBook, BaaderW24}, CTL \citep{ClarkeE08,clarke:model-checking} and LTL \citep{pnueli:ltl}.
Due to the higher expressive power of these logics, %
the constructive contraction operators in these settings differ from the classical ones, though preferences relations still remain the main component to obtain such rational constructive operators \citep{jandson:towards-contraction, jandson:sans-compactness, kr25:effectiveAGMContraction}. %
In this sense, preference relations form the backbone of rational operators, even beyond the classical cases. %
Such preference relations must satisfy some innocuous conditions to guarantee rationality of the contraction operators.

Belief change investigates and provides solutions for  problems intrinsically connected with the flux of knowledge such as \textit{ontology repair, reasoning in the presence of inconsistency, data integration} and \textit{non-monotonic reasoning}, to cite a few. 
This makes belief change a central component to successfully support and implement various forms of knowledge dynamics and their respective applications. 
Therefore, it is paramount to understand  the computational powers and limitations of belief change.
Despite this central role of belief change in handling knowledge dynamics,
the computational aspects of belief change are almost an untouched subject.
The few works that look at it concentrate on specific change operators within classical propositional logics or Horn \citep{Nebel1998, EiterG92, schwind:compute-iterated, SB22}.
Only recently, this question started to receive attention for logics beyond the classical spectrum, and revealed dire disruptive negative results.
\citet{guerra2018} have investigated some decidability problems regarding the traditional partial meet contraction operators in LTL, and have shown that deciding whether a given contraction operator is a partial meet operator is undecidable.
\citet{Souza22}  investigated the problem of reasoning about belief change when the underlying knowledge is specified in CTL, and revision operators are encapsulated in  the object language of preference logics. 
\citet{kr25:effectiveAGMContraction}
have shown that there are uncountably many uncomputable contraction operators in LTL and other non-classical logics.
They framed the precise class of computable contraction operators in LTL, using \emph{B\"uchi automata}, and \emph{B\"uchi-Mealy automata} to specify preference relations.

Given the central role of contraction in supporting and obtaining other forms of belief change operators,
contraction is the ideal candidate for studying the computational aspects of belief change, 
as by effectively computing contraction, we lay the foundation to  construct other forms of change operators. 
Following the results by \citet{kr25:effectiveAGMContraction}, using 
\buchi-Mealy automata to specify preference relation in LTL, %
we consider the problem of verifying whether a preference relation specified as a given \buchi-Mealy automaton satisfies the conditions required to secure rationality of contraction.
We show both positive and negative results,
and discuss how to address the negative results.

\textit{Contributions and roadmap.}
\cref{sec:preliminaries_ltl} presents the basic notation and definitions used throughout the paper, and briefly reviews LTL and \buchi\ automata,
while
\cref{sec:AGM_contraction} reviews the rationality postulates of belief contraction,
and the respective preference relations and rational constructions  for LTL. %
The subsequent two sections contain our main contributions:
\begin{itemize}
\item \Cref{sec:morrorring_and_cuts} investigates the decidability of the two central conditions for rationality.
  One of them, \emph{mirroring}, is decidable (\cref{thm:mirroring-dec}),
  while the other condition, \emph{maximal cut}, which ensures success of the contraction,
  is undecidable (\cref{thm:max-cut-undec}).
  The same applies to a stricter form of maximal cut, \emph{well-foundedness}.
  We explain why this rules out the possibility of certain representation results.
\item In \cref{sec:constructions},
  we show how %
  to achieve effective LTL contraction despite these issues,
    by presenting constructive families of preference relations
    that guarantee \emph{mirroring} and \emph{maximal cut} by design. Such families of relations not only provide a way to compute contraction for LTL, but also provide richer alternatives of preference relations for the classical settings.
\end{itemize}
Finally, \cref{sec:conclusion} discusses our results, their impact on the field and potential future research paths.

\ifextended\else Proofs of our results are available in an extended version uploaded as supplementary material.\fi

\section{LTL and B\"{u}chi Automata}
\label{sec:preliminaries_ltl}

We introduce some notations and definitions used throughout the paper, and briefly review \emph{linear temporal logic}, LTL for short.
The power set of a set $X$ is denoted $\powerset{X}$. 
For the remainder of the paper, we fix a finite, non-empty set of atomic propositions $\AP$. 
\begin{definition}[LTL Formulae]

	\emph{The formulae of LTL} are generated by the grammar
	\[
	\varphi \mathrel{::=} \bot \ |\ p \ |\ \lnot \varphi \ |\ \varphi \lor \varphi \ |\ \Next\varphi \ |\ \varphi \Until \varphi\ ,
	\]%
	where $p$~ranges over~$\AP$.
	$\Fm_\LTL$ denotes the set of all LTL formulae.
\end{definition}

In LTL, time is interpreted as a single timeline that unfolds to the future.
Formally, timelines are \emph{traces}, infinite sequences $\trace =a_0 a_1 \cdots $ where each $a_i \in \powerset{\AP}$ is the set of atomic propositions that holds at the time step $i$.
We write $\tau^i$ for the trace~$a_i a_{i+1} \cdots$.
The semantic of LTL is given in terms of Kripke structures and their traces.   

\begin{definition}[Kripke Structure]
	A \emph{Kripke structure} $M = (S, I, T, \lambda)$ consists of
	a finite set of states~$S$;
	a non-empty set of initial states~$I \subseteq S$;
	a left-total transition relation~$T \subseteq S \times S$, i.e., for all $s \in S$ there exists $s' \in S$ with $(s,s')\in T$;
	and a labeling~$\lambda : S \to \powerset{\AP}$ of states with sets of atomic propositions.
\end{definition}

A trace of a Kripke structure $M$ is a sequence
${\tau = \lambda(s_0)\lambda(s_1)\lambda(s_2)\cdots}$ with
$s_0 \in I$, and for all $i \geq 0$,
$s_i \in S$ and $(s_i, s_{i+1})\in T$.
The set of all traces of a Kripke structure $M$ is given by $\Traces{M}$.

The satisfaction relation between Kripke structures and LTL formulae is defined in terms of the satisfaction between the Kripke structure's traces and LTL formulae.
\goodbreak

\begin{definition}[Satisfaction]\label{def:sattraces}
	The \emph{satisfaction relation} is the least relation ${\models} \subseteq \powerset{\AP}^\omega \times \Fm_\LTL$ between traces and LTL formulae such that, for all $\tau = a_0a_1\cdots \in \powerset{\AP}^\omega$:  %
	\newcommand\myiff{\mbox{iff }}
	\begin{align*}
		\begin{array}{lclllcl}
			\tau \not\models \bot &&& ~ \qquad\quad ~
			& \tau \models \varphi_1 \lor \varphi_2 &\myiff & \tau \models \varphi_1 \text{ or } \tau\models\varphi_2\\
			\tau \models p &\myiff & p\in a_0
			&& \tau \models \Next\varphi &\myiff & \tau^1 \models \varphi\\
			\tau \models \lnot\varphi &\myiff & \pi \not\models \varphi
			&& \tau \models \varphi_1 \Until \varphi_2 &\myiff &\text{ there exists } i \geq 0 \text{ s.t. } \tau^i \models \varphi_2\\
						&&&&&& \text{ and for all } j < i, \tau^j \models \varphi_1
		\end{array}
	\end{align*}
\end{definition}
The set of all traces satisfying an LTL formula $\varphi$ is denoted $\tracesof{\varphi}$.
Traces that do not satisfy~$\varphi$ are called \emph{contra-traces} of $\varphi$.
As a trace does not satisfy a formula  $\varphi$ iff it satisfies its negation $\neg \varphi$, the set of all contra-traces of $\varphi$ is $\tracesof{\neg \varphi}$. 
A Kripke structure $M$ satisfies a formula $\varphi$, denoted ${M \models \varphi}$, iff all traces of $M$ satisfy $\varphi$.
$M$ satisfies a set $X$ of formulae, $M \models X$, iff $M \models \varphi$ for all $\varphi \in X$. 

From the satisfaction relation, we  define a consequence operator $\CnFun_{\LTL}: \powerset{\FmLTL} \to \powerset{\FmLTL}$ which  maps each set of formulae $X$ to all its consequences.

\begin{definition}[Consequence Operator]
	The \emph{consequence operator} $\CnFun_\LTL$ maps each set $X$ of LTL formulae
	to the set of all formulae $\psi$,
	such that for all Kripke structures $M$,
	if~$M \models X$ then also $M \models \psi$.
\end{definition}

A \emph{theory} is a logically closed set of formulae $\kb$, that is,  $\CnLTL{\kb} = \kb$.
The expansion of a theory~$\kb$ by a formula $\varphi$ is the theory $\kb + \varphi := \CnLTL{\kb \cup \{\varphi\}}$.
A theory $\kb$ is \emph{consistent} if $\kb\neq \FmLTL$,
and it is \emph{complete} if for all formulae $\varphi \notin \kb$, we have $\kb + \varphi = \FmLTL$.
The set of all complete consistent theories is denoted as $\CCT$.

B\"uchi automata are finite automata widely used in formal specification and verification of systems,
especially in LTL model checking %
\citep{clarke:model-checking},
as well as in planning~\citep{GiacomoV99,PatriziLGG11} when goals are specified in LTL.  %
\begin{definition}[B\"uchi Automata]%
	A \emph{B\"uchi automaton}~$A = (Q, \Sigma, \Delta, Q_0, R)$
	consists of
	a finite set of states~$Q$;
	a finite, nonempty alphabet $\Sigma$ (whose elements are called \emph{letters});
	a transition relation $\Delta \subseteq Q \times \Sigma \times Q$;
	a set of initial states $Q_0 \subseteq Q$;
	and a set of \emph{recurrence states} $R \subseteq Q$.
\end{definition}

A B\"uchi automaton accepts an infinite word over the alphabet $\Sigma$
if the automaton visits a recurrence state infinitely often while reading the word.
Formally, an infinite word is a sequence $a_0 a_1 \ldots$ with $a_i\in \Sigma$ for all~$i$.
For a finite word $\rho = a_0\ldots a_n$, with $n \geq 0$, let $\rho^\omega$~denote the infinite word derived by infinite repetition of $\rho$
(we assume $\cdot\,^\omega$ binds stronger than concatenation).
The set of all infinite words is denoted by $\Sigma^\omega$.
An infinite word $a_0a_1a_2\ldots \in \Sigma^\omega$ is \emph{accepted} by a B\"uchi automaton $A = (Q, \Sigma, \Delta, Q_0, R)$ if
there exists a sequence $q_0,q_1,q_2,\ldots$ of states $q_i\in Q$ such that
$q_0\in Q_0$ is an initial state,
for all~$i$ we have that $(q_i, a_i, q_{i+1}) \in \Delta$
and there are infinitely many $i\in \mathbb{N}$ with $q_i \in R$.
The set~$\lang{A}$ of all accepted words is the \emph{language} of~$A$.

Emptiness of a B\"uchi automaton's language is decidable.
Further, B\"uchi automata for the union, intersection and complement of the languages of given B\"uchi automata can be effectively constructed~\citep{buchi:automata}.
Unless otherwise noted, we always consider B\"uchi automata over the alphabet $\Sigma = \powerset{\AP}$,
where letters are sets of atomic propositions and infinite words are traces.
The automata-theoretic treatment of LTL is based on the following result:

\begin{proposition}[\citet{clarke:model-checking}] \label{prop:phi_kripke_to_buchi}%
	\label{prop:ltl-kripke-buchi}%
	For each LTL formula $\varphi$ and Kripke structure $M$,
	there exist  B\"uchi automata~$\ltlbuchi{\varphi}$ and~$\kripkebuchi{M}$
	that accept precisely $\lang{\ltlbuchi{\varphi}} = \tracesof{\varphi}$ resp.\ $\lang{\kripkebuchi{M}} = \Traces{M}$.
	
\end{proposition}

We use \buchi\ automata as structures to represent an agent's beliefs.
In the original AGM paradigm of belief change \citep{gardenfors:flux, alchourron:partial-meet}, an agent's corpus of beliefs is represented as a theory.
However, \buchi\ automata accept traces, rather than formulae.
This bridge is accomplished by the notion of support, which takes all formulae satisfied by all traces accepted by a \buchi\ automaton.
\Citet{kr25:effectiveAGMContraction} have shown that the support is indeed a theory.

\begin{definition}[Support \citep{kr25:effectiveAGMContraction}]\label{def:support}
	The \emph{support} of a B\"{u}chi automaton $A$ is the set %
	$
	\support{A} := \{\, \varphi \in \Fm_\LTL \mid \forall \pi\in\lang{A}\,.\, \pi \models \varphi \,\}
	$.
	If $\varphi \in \support{A}$, we say that $A$ \emph{supports} $\varphi$.
\end{definition}
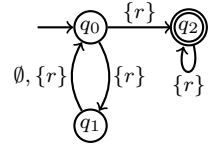
\begin{wrapfigure}[7]{r}{0.22\textwidth}
\centering
\vspace*{-2em}
\begin{tikzpicture}[thick,inner sep=1,node distance=13mm,font=\scriptsize]
  \node[draw,circle] (q0) {$q_0$};
  \node[draw,circle,below of=q0] (q1) {$q_1$};
  \node[draw,double,circle,right of=q0] (q2) {$q_2$};
  \draw[<-] (q0) -- ++(left:7mm);
  \draw[->] (q0) edge[bend left] node[auto] {$\{r\}$} (q1);
  \draw[->] (q1) edge[bend left] node[auto] {$\emptyset,\{r\}$} (q0);
  \draw[->] (q0) -- node[auto] {$\{r\}$} (q2);
  \draw[->] (q2) edge[loop below] node[auto] {$\{r\}$} ();
\end{tikzpicture}
\captionsetup{hangindent=0pt,indention=0pt,labelfont=bf,format=plain}
\caption{B\"uchi automaton with initial state~$q_0$ and recurrence state~$q_2$.}
\label{fig:nba-example}
\end{wrapfigure}%
The space of all theories represented by \buchi\ automata is called \buchi\ excerpt, denoted $\excBuchi$.
Notably, some theories cannot be represented by B\"uchi automata.
In non-finitary logics like LTL, this is unavoidable for any method of finite representation~\cite{kr25:effectiveAGMContraction}.
\begin{example}
  Let the atomic proposition~$r$ denote \quot{it rains}.
  \Cref{fig:nba-example} shows a B\"uchi automaton~$A$,
  which supports the formulae $r$ (\quot{it rains today}), $\Next[2] r$ (\quot{it will rain in two days}), $\Next[4] r$ (\quot{it will rain in 4~days}), etc.
  However, $A$~does not support~$\Globally r$ (\quot{it rains every day}) because it accepts the trace~$\{r\}\,\emptyset\,\{r\}^\omega$.
\end{example}

\section{Belief Contraction in LTL}\label{sec:AGM_contraction}
\label{sec:ltl-contraction}

The field of belief change is founded on the AGM paradigm \citep{alchourron:partial-meet, gardenfors:flux}, originally devised for classical logics. 
In the original, a logic is treated as a pair $(\Fm, \CnFun)$ where $\Fm$ denotes the language of the logic, and $\CnFun$ is a consequence operator mapping each set of formulae to the set of all its entailments. 
The epistemic state of an agent is represented as a theory.  
A contraction function on a theory~$\kb$\footnote{Note that we are following the traditional local definition of contraction function rather than the global one. As thoughtfully studied by \citet{hansson:belief-dynamics}, global functions do not bring any benefit on one-shot contractions.} is a function ${\dotmin}: \Fm \to \powerset{\Fm}$ that,
given an obsolete piece of information~$\varphi$,
produces a theory  which does not entail~$\varphi$.
Contraction functions are subject to the following
rationality postulates~\citep{gardenfors:flux}:
\begin{tabbing}
kkkkk\= hhhhhghjhkhkjkjhjkhjhjjhk\= hjhjhhhhhhh\= kkkkk\= hhhhhghjhkhkjkjhjkhjhjjhjhjhjk\= hjhjhhh\kill
\cp[1] \>  $\kb \dotmin \varphi = \Cn{\kb \dotmin \varphi}$ \> (closure)
\> \cp[2]\> $\kb \dotmin \varphi \subseteq \kb$ \> (inclusion) \\
\cp[3]\>  If $\varphi \not \in \kb$ then $\kb \dotmin \varphi = \kb$ \> (vacuity)
\> \cp[4] \> If $\varphi \not \in \Cn{\emptyset}$ then $\varphi \not \in \kb \dotmin \varphi$ \> (success)\\
\cp[5]\> $\kb \subseteq (\kb \dotmin \varphi) + \varphi$ \> (recovery)
\> \cp[6]\> If $\varphi\equiv\psi$ then $\kb \dotmin \varphi = \kb \dotmin \psi$ \>  (extensionality)\\[1em]
\cp[7]\> $(\kb \dotmin\varphi)\cap(\kb \dotmin \psi)\subseteq \kb\dotmin{(\varphi\wedge \psi)}$ \>
\> \cp[8]\> If $\varphi\not\in \kb\dotmin{(\varphi\wedge \psi)}$ then
$\kb\dotmin{(\varphi\wedge \psi)}\subseteq \kb\dotmin\varphi$ \>
\end{tabbing}

For a detailed discussion on the rationale of these postulates,  see \citep{alchourron:partial-meet,gardenfors:flux,hansson:belief-dynamics}.
A contraction function that satisfies \cp[1] to \cp[6] is called a \emph{rational} contraction function.
If a  contraction function satisfies all the eight rationality postulates,  we say that it is \emph{fully rational}. 

Several classes of (fully) rational AGM contraction functions were proposed on classical logics \citep{hansson:belief-dynamics}. 
In more expressive, non-classical logics such as LTL, however, such classical classes of functions are too weak to capture recovery. To embrace more expressive logics,  \citet{jandson:towards-contraction} have proposed new classes of (fully) rational contraction functions. 
Upon these, \citet{kr25:effectiveAGMContraction} %
characterised precisely the class of all rational AGM contraction functions on LTL theories in the B\"uchi excerpt, called \buchi\ contraction functions.

B\"uchi contraction functions work by selecting the most plausible  contra-traces of the LTL formula to be contracted, and incorporating these traces.  %
The selection is realised by a \emph{B\"uchi choice function}, guaranteeing that the set of selected contra-traces is represented as a \buchi\ automaton.

\begin{definition}[B\"uchi Choice Functions]\label{def:choice}
	A \emph{B\"uchi choice function}~$\gamma$ maps each LTL formula to a single B\"uchi automaton, such that
	for all LTL formulae~$\varphi$ and~$\psi$,
	\begin{description}
		\item[\propdef{BF1}] the language accepted by $\gamma(\varphi)$ is non-empty;
		\item[\propdef{BF2}] $\gamma(\varphi)$ supports $\lnot\varphi$, if $\varphi$ is not a tautology; and
		\item[\propdef{BF3}] $\gamma(\varphi)$ and $\gamma(\psi)$ accept the same language, if $\varphi \equiv \psi$.
	\end{description} 
\end{definition}

As the purpose is to relinquish a formula $\varphi$, the choice function must select only among the contra-traces of $\varphi$~\prop{BF2}, except for tautologies, which cannot be removed.
To be successful, at least one contra-trace of $\varphi$ must be selected~\prop{BF1}.
The choice is not syntax-sensitive~\prop{BF3}.
The contraction corresponds to assimilating all selected traces into the current \buchi\ automaton, as long as the formula is not a tautology, in which case the theory remains untouched.

\begin{definition}[B\"uchi Contraction Functions]
	\label{def:contractp}
	Let $\kb$ be a theory in the B\"uchi excerpt
	and let
	$\gamma$ be a B\"uchi choice function.   %
	The \emph{B\"uchi Contraction Function} on $\kb$ induced by $\gamma$ is the function %
	$\dotmin[\gamma]$ such that
	$\kb \dotmin[\gamma] \varphi = \kb \cap \support{\gamma(\varphi)}$ if $\varphi \notin \Cn{\emptyset}$ and $\varphi \in \kb$,
	or $\kb \dotmin[\gamma] \varphi = \kb$ otherwise.
\end{definition}

\begin{theorem}[\citet{kr25:effectiveAGMContraction}]
	\label{cor:comp-basic-contraction}
	A contraction function $\dotmin$ in the \buchi\ excerpt 
	is rational if and only if  ${\dotmin}$ is a \buchi\ contraction function.
\end{theorem}

For the supplementary postulates \cp[7] and \cp[8], the choice function has to be constrained further.
The main idea is that the choices must be based on epistemic preferences an agent has about its beliefs, with some beliefs being more plausible than others.
Such preferences are realised by a binary relation~${\pref}$ between the traces, where $\trace \pref \trace'$ indicates that $\trace'$~is more plausible than~$\trace$%
\footnote{
	As most of our results are founded on \citet{jandson:towards-contraction}, we follow their interpretation on preferences %
	as ``the higher, the better'', unlike the classical interpretation founded on the entrenchment notion ``the more grounded, the better''.}%
The choice then boils down to choosing the most plausible contra-traces of the formula~$\varphi$ being relinquished,
that is, $\lang{\gamma_{\pref}(\varphi)} = \max_{\pref} \tracesof{\neg \varphi}$.
The preference relation is subject to two central conditions:

\begin{description}
	\item[\propdef{Maximal Cut}] $\max_{\pref} \tracesof{\neg \varphi} \neq \emptyset$, if $\varphi$ is not a tautology.
	
	\item[\propdef{Mirroring}] If $\trace_1 \npref \trace_2$ and $\trace_2  \npref \trace_1$ but $\trace_1 \pref \trace_3$ then  $\trace_2 \pref \trace_3$.
\end{description}

 \prop{Maximal Cut} guarantees the success of the contraction, by ensuring the existence of maximally plausible contra-traces for each formula.
 The second property \prop{Mirroring} enforces that if two traces are incomparable, then the preferences over them must mimic one another. 
 To obtain a \buchi\ choice function, it still necessary to guarantee that $\max_{\pref}\tracesof{\neg \varphi}$ is recognised by a \buchi\ automaton.
 To this end, 
 \citet{kr25:effectiveAGMContraction} proposed the use of \emph{B\"uchi-Mealy automata} to represent such  relations.

\begin{definition}[B\"uchi-Mealy Automata]
	A \emph{B\"uchi-Mealy automaton} (BMA, for short) is a B\"uchi automaton over the alphabet~$\Sigma_\textnormal{BM} = \powerset{\AP} \times \powerset{\AP}$.
\end{definition}

A \buchi-Mealy automaton~$B$ accepts infinite sequences of pairs $(a_1, b_1) \cdots (a_i, b_i) \cdots$.
Such an infinite sequence encodes a pair of traces $(\trace_a, \trace_b)$, where 
$\trace_a = a_1 \cdots a_i \cdots $ and 
$\trace_b = b_1 \cdots b_i \cdots $. 
Hence, a BMA~$B$ recognises the binary relation %
$\rel{B} = \{ (a_1 \ldots, b_1 \ldots ) \in \Sigma^\omega \times \Sigma^\omega \mid (a_1, b_1) \cdots \in \lang{B} \}$.

\Citet{kr25:effectiveAGMContraction} have shown that the relation~$\rel{B}$ enjoys a key property:
for each LTL formula~$\varphi$,
the most plausible traces of a formula~$\varphi$ wrt.~$\rel{B}$ are recognised by a \buchi\ automaton,
which can be effectively computed.
In fact, this property, together with \prop{Maximal Cut}, guarantees that the preference relation~$\rel{B}$ of a BMA~$B$ induces a \buchi\ choice function~$\gamma_{B}$
with $\mathcal{L}\big(\gamma_{B}(\varphi)\big)  = \max_{\rel{B}} \tracesof{\neg \varphi}$.
The respective \buchi\ contraction function, denoted by~$\dotmin[\rel{B}]$, is guaranteed to be computable.
If $\rel{B}$~additionally satisfies \prop{Mirroring}, the contraction function~$\dotmin[\rel{B}]$~is fully rational.

\begin{theorem}[\citet{kr25:effectiveAGMContraction}]\label{theo:fullybuchi}
If the preference relation $\rel{B}$ of a BMA satisfies both \prop{Maximal Cut} and \prop{Mirroring},
the induced \buchi\ contraction function is fully rational and computable.
\end{theorem}
\begin{example}
  \label{ex:bma-preferences}
  Let the proposition~$r$ once again denote that \quot{it rains}.
  \Cref{fig:bma-example} shows an example of a B\"uchi-Mealy automaton~$B$.
  By convention, we write $x/y$ on the edges in place of the pair~$(x,y)$.

  As shown to the right of the figure,%
  the relation~$\rel{B}$ considers the trace~$\emptyset\{r\}\emptyset^\omega$
  (in which it does\pagebreak[3] not rain today, but tomorrow)
  less plausible than the trace~$\{r\}\emptyset^\omega$
  (in which it rains today),
  due to the corresponding run~$q_0\, q_2^\omega$.
  Similarly, the preference on the second line
  (preferring a trace in which it rains tomorrow over one in which it rains only in two days)
  is witnessed by the run~$q_0\,q_1\,q_2^\omega$.
  There is no preference between the traces in the third line (rain in two vs.\ in three days),
  as there is no corresponding run (after the partial run~$q_0\,q_1\,q_3\,q_4$, the automaton~$B$ is ``stuck'').
  Finally, $B$~considers the trace $\emptyset^\omega$ (in which it never rains)
  less plausible than~$\emptyset\emptyset\emptyset\{r\}\emptyset^\omega$,
  by the run~$q_0\,q_1\,q_3\,q_3\,q_4^\omega$.
  It can be shown that $\rel{B}$ satisfies \prop{Mirroring} and \prop{Maximal Cut}.
  We return to this point in \cref{ex:next-days-rain-ranking}.
\end{example}

\begin{figure}[t]
\begin{minipage}{0.5\textwidth}
\begin{tikzpicture}[thick,font=\footnotesize,node distance=2cm,inner sep=1]
  \node[draw,circle] (q0) {$q_0$};
  \node[draw,circle,right of=q0] (q1) {$q_1$};
  \node[draw,circle,double,right of=q1] (q2) {$q_2$};
  \node[draw,circle,below=1cm of q1] (q3) {$q_3$};
  \node[draw,circle,double,right of=q3] (q4) {$q_4$};

  \draw[<-] (q0) -- ++(left:7mm);
  \draw[->] (q0) -- node[below]{$\emptyset/\emptyset$} (q1);
  \draw[->] (q1) -- node[below]{$\emptyset/\{r\}$} (q2);
  \draw[->] (q1) -- node[left]{$\emptyset/\emptyset$} (q3);
  \draw[->] (q3) -- node[below]{$\emptyset/\{r\}$} (q4);
  \draw[->] (q0) edge[bend left] node[above]{$\emptyset/\{r\}$} (q2);
  \draw[->] (q3) edge[loop left] node[auto]{$\emptyset/\emptyset$} ();
  \draw[->] (q2) edge[loop right] node[auto,align=left]{$\emptyset/\emptyset$,\ \ $\emptyset/\{r\}$\\$\{r\}/\emptyset$,\ \ $\{r\}/\{r\}$} (); %
  \draw[->] (q4) edge[loop right] node[auto,align=left]{$\emptyset/\emptyset,$\\$\emptyset/\{r\}$} ();
\end{tikzpicture}
\end{minipage}
\hfill
\begin{minipage}{0.45\textwidth}
  \begin{align*}
    \emptyset\ \{r\}\ \emptyset\ \emptyset\ \emptyset \cdots && \pref && \{r\}\ \emptyset\ \emptyset\ \emptyset\ \emptyset \cdots\\
    \emptyset\ \emptyset\ \{r\}\ \emptyset\ \emptyset \cdots && \pref && \emptyset\ \{r\}\ \emptyset\ \emptyset\ \emptyset \cdots\\
    \emptyset\ \emptyset\ \emptyset\ \{r\}\ \emptyset \cdots && \#    && \emptyset\ \emptyset\ \{r\}\ \emptyset\ \emptyset \cdots\\
    \emptyset\ \emptyset\ \emptyset\ \emptyset\ \emptyset \cdots && \pref && \emptyset\ \emptyset\ \emptyset\ \{r\}\ \emptyset \cdots
  \end{align*}
\end{minipage}
\caption{A B\"uchi-Mealy automaton~$B$ (on the left), along with some examples of preferences between traces,
  as given by the relation~${\pref} = \rel{B}$.
  The notation~$\tau_1 \mathrel{\#} \tau_2$ indicates that both $\tau_1\not\pref \tau_2$ and $\tau_2\not\pref\tau_1$ hold.}
\label{fig:bma-example}
\end{figure}
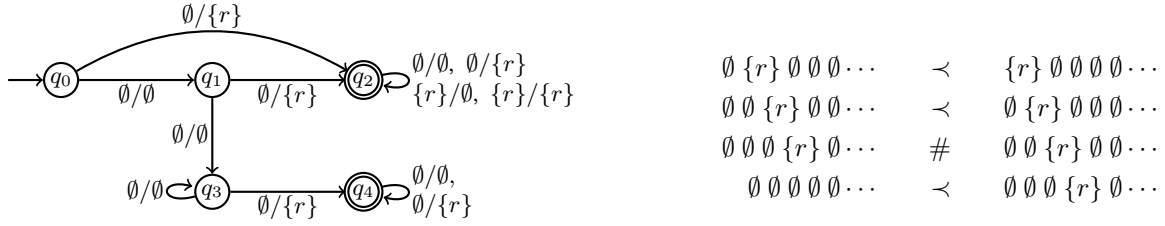

For an effective realization of fully rational LTL belief contraction based on epistemic preferences represented as BMA,
it is thus necessary to ensure that the particular BMA relation satisfies \prop{Mirroring} and \prop{Maximal Cut}.
While it is possible to manually verify these properties for individual BMA relations,
allowing for the construction of bespoke contraction operators,
a general framework of belief contraction requires a general solution to this problem.
Two main alternatives are:
(i)~identify general constructions guaranteeing that the underlying relation can be encoded as a BMA and satisfies \prop{Maximal Cut} and \prop{Mirroring},
or
(ii)~provide devices capable of deciding whether the relation of a given BMA satisfies \prop{Maximal Cut} and \prop{Mirroring}.
In the next section, we show some fundamental limitations of such a general solution.

\section{On Deciding Mirroring and Maximal Cut}\label{sec:morrorring_and_cuts}

We investigate the \emph{decidability} of the two fundamental properties of BMA relations that guarantee full rationality: \prop{Mirroring} and \prop{Maximal Cut}.
We show that \prop{Mirroring} is decidable, but \prop{Maximal Cut} is undecidable.
As we see later on, this result precludes the possibility of a constructive representation result for fully rational BMA relations.
In addition to \prop{Maximal Cut}, we also investigate the conceptually simpler but strictly stronger property \emph{well-foundedness}.
While this shift does not improve the decidability situation, well-foundedness forms the basis for a set of pragmatic constructions that guarantee full rationality (\cref{sec:constructions}).

\subsection{Mirroring}
\label{sec:decide-mirroring}
\begin{toappendix}
  \subsection{Proofs for \cref{sec:decide-mirroring}}
\end{toappendix}
We begin our investigation with the property \prop{Mirroring}.
While the definition of this property in~\cref{sec:ltl-contraction} is given in a \emph{point-wise fashion}
(by quantifying over individual points, i.e., traces),
the property can be equivalently stated \emph{algebraically}, using standard operators from relational algebra.
\begin{toappendix}
  To begin, we repeat some standard definitions concerning relations.
    Let~$R, S\subseteq X\times X$ denote relations over a set~$X$.
    The \emph{complement}~$\overline{R}$ and \emph{inverse}~$R^{-1}$ of~$R$ are defined as
    \begin{align*}
      \overline{R} &= \{(a,b)\in X\times X\mid (a,b)\not\in R\}\\
      R^{-1} &= \{(a,b)\in X\times X\mid (b,a)\in R\}\ ,
    \end{align*}
    and the \emph{composition}~$R\circ S$ of $R$~and~$S$ as
    \[R\circ S = \{(a,c)\in X\times X\mid b\in X,(a,b)\in R,(b,c)\in S\}.\]
  We use these constructs in our algebraic formulation of \prop{Mirroring}:
\end{toappendix}
\begin{lemmarep}
  \label{lem:mirroring-algebraically}
  A relation~$R$ satisfies \prop{Mirroring}
  if and only if the inclusion $({\overline{R}} \cap ({\overline{R})^{-1}}) \circ {R} \ \ \subseteq \ \ {R}$ holds,
  where $\overline{\,\cdot\,}$ denotes the complement, $(\cdot)^{-1}$ denotes the inverse relation, and $\circ$ denotes relational composition.
\end{lemmarep}
\begin{proof}
  For the first direction, let us assume that $R$ satisfies the inclusion above. Let $x_1,x_2,x_3\in X$ such that $(x_1,x_2)\not\in R$, $(x_2,x_1)\not\in R$ and $(x_1,x_3)\in R$.
  In that case, we have:
  \begin{align*}
    &(x_1,x_2)\not\in R \land (x_2,x_1)\not\in R \land (x_1,x_3)\in R\\
    \null\Leftrightarrow\null &(x_2,x_1)\in(\overline{R})^{-1} \land (x_2,x_1)\in\overline{R} \land (x_1,x_3)\in R\\
    \null\Leftrightarrow\null &(x_2,x_1)\in {\overline{R}} \cap (\overline{R})^{-1} \land (x_1,x_3)\in R\\
    \null\Rightarrow\null &(x_2,x_3)\in ({\overline{R}} \cap ({\overline{R})^{-1}}) \circ {R}\\
    \null\Rightarrow\null &(x_2,x_3)\in R¨
  \end{align*}
  Hence, $R$~satisifies \prop{Mirroring}.

  For the other direction, let us assume that $R$~satisfies \prop{Mirroring} and let $x_2,x_3\in X$.
  We have:
  \begin{align*}
    &(x_2,x_3)\in ({\overline{R}} \cap ({\overline{R})^{-1}}) \circ {R}\\
    \null\Leftrightarrow\null &\exists x_1\in X \,.\, (x_2,x_1)\in {\overline{R}} \cap (\overline{R})^{-1} \land (x_1,x_3)\in R\\
    \null\Leftrightarrow\null &\exists x_1\in X \,.\, (x_1,x_2)\not\in R \land (x_2,x_1)\not\in R \land (x_1,x_3)\in R\\
    \null\Rightarrow\null &(x_2,x_3)\in R
  \end{align*}
  Therefore, the inclusion is satisfied.
\end{proof}
This algebraic characterisation is useful to establish decidability,
due to the following observation:

\begin{lemmarep}
  \label{lem:bma-closed-ops}
  B\"uchi-Mealy automata are effectively closed under union, intersection, complement, inverse relation and relational composition.
\end{lemmarep}
\begin{proof}
  B\"uchi-Mealy automata are a special case of B\"uchi automata, so constructions for the union, intersection and complement of B\"uchi automata can be applied~\cite{buchi:automata,buchi:s1s}.

  Given a BMA $B=(Q, \Sigma_\textnormal{BM}, \Delta, Q_0, R)$,
  we construct a BMA~$B'$ for the inverse relation, i.e., a BMA that satisfies $\rel{B'}=\rel{B}^{-1}$,
  by setting $B' := (Q, \Sigma_\textnormal{BM}, \Delta', Q_0, R)$,
  where $\Delta' =\{\,(q,(b,a),q') \mid (q,(a,b),q')\in\Delta\,\}$.

  Given two BMA $B^1=(Q^1, \Sigma_\textnormal{BM}, \Delta^1, Q_0^1, R^1)$ and $B^2=(Q^2, \Sigma_\textnormal{BM}, \Delta^2, Q_0^2, R^2)$,
  we construct a BMA~$B^\circ$ for the relational composition, i.e., a BMA that satisfies $\rel{B^\circ}=\rel{B^1}\circ\rel{B^2}$,
  by setting \[B^\circ := \big(Q^1\times Q^2 \times \{1,2\}, \Sigma_\textnormal{BM}, \Delta^\circ, Q_0^1\times Q_0^2 \times \{1\}, (R^1\times Q^2 \times \{1\}) \cup (Q^1 \times R^2 \times \{2\})\big)\ ,\]
  where $((q_1,q_2,k),(a,c),(q'_1,q'_2,k'))\in \Delta^\circ$ if there is some $b\in\Sigma$ such that $(q_1,(a,b),q'_1)\in\Delta^1$ and $(q_2,(b,c),q'_2)\in\Delta^2$,
  and we have $k'=3-k$ if $q_k\in R_k$, otherwise $k'=k$.
\end{proof}

Moreover, the inclusion between B\"uchi-Mealy automata is decidable, as a special case of B\"uchi automata inclusion~\citep{buchi:automata}.
We conclude:
\begin{theorem}
  \label{thm:mirroring-dec}
  The BMA mirroring problem, as given below, is decidable in \textsc{ExpTime}.
  \decproblem{A B\"uchi-Mealy automaton $B$}{$\rel{B}$ satisfies \prop{Mirroring}}
\end{theorem}
\begin{proof}
  It suffices to construct a BMA for the left-hand side of the inclusion in~\cref{lem:mirroring-algebraically}, %
  where $R=\rel{B}$,
  and decide (in \textsc{ExpTime}) the inclusion between this BMA and~$B$.
  The dominating factor is the complementation of~$B$, which takes exponential time and produces a complement automaton that may be exponentially larger than~$B$~\citep{schewe:tight-buchi-complementation}.
  The other operations are at worst quadratic.
\end{proof}

\subsection{On Maximal Cut and Well-Foundedness}
\label{sec:maxcut-vs-wellfounded}
\begin{toappendix}
  \subsection{Proofs for \cref{sec:maxcut-vs-wellfounded}}
\end{toappendix}
Before we show undecidability of \prop{Maximal Cut} in~\cref{sec:maximal-cut},
let us spend a moment studying the connection (and distinctions) between \prop{Maximal Cut} and another, more classical property of relations: \emph{well-foundedness}.
We begin our comparison with an equivalent reformulation of \prop{Maximal Cut}, which serves to highlight the connection between the properties:
\begin{description}
\item[(Maximal Cut)] For every nonempty set~$X \subseteq \Sigma^\omega$,
  {\color{blue}
    if there exists~$\varphi\in\Fm_\LTL$ with $X=\tracesof{\varphi}$,
  } then $\max_{\pref}(X) \neq \emptyset$.
\end{description}
\begin{toappendix}
  Let us briefly remark on the equivalence between the original definition of \prop{Maximal Cut}
  and the formulation given at the start of \cref{sec:maxcut-vs-wellfounded}.
  \begin{observation}
    The relation~$\pref$ satisfies \prop{Maximal Cut}
    if and only if
    for every nonempty set~$X\subseteq \Sigma^\omega$,
    if there exists~$\varphi\in\Fm_\LTL$ with $X=\tracesof{\varphi}$,
    then $\max_{\pref}(X) \neq \emptyset$.
  \end{observation}
  \begin{proof}
    The reformulation is mostly straightforward.
    The only key difference is that we consider the traces of~$\varphi$,
    i.e., the set~$\tracesof{\varphi}$,
    rather than its contra-traces~$\tracesof{\lnot\varphi}$.
    However, as we quantify over all LTL formulae~$\varphi$,
    and LTL enjoys classical negation,
    this is equivalent:
    We derive non-emptiness of~$\tracesof{\lnot\varphi}$ by applying the reformulated condition to~$\varphi' :\equiv \lnot\varphi$,
    and vice versa.
  \end{proof}
\end{toappendix}
As {\color{blue}highlighted}, this property requires a careful determination which sets~$X$ of traces can be represented by an LTL formula~$\varphi$.
By contrast, \emph{well-foundedness} omits this restriction of the set~$X$:
\begin{description}
\item[\propdef{Well-Founded}] For every nonempty set~$X\subseteq \Sigma^\omega$,
it holds that $\max_{\pref}(X) \neq \emptyset$.\\
\emph{Or equivalently:} There exists no infinite sequence $\tau_0 \pref \tau_1 \pref \ldots$,
  with $\tau_i \in \Sigma^\omega$ for all~$i\in\N$.%
  \footnote{Due to this formulation, \prop{Well-Founded} is sometimes called the \emph{ascending chain condition} (ACC).}
\end{description}
Clearly, \prop{Well-Founded} implies \prop{Maximal Cut}.
The reverse implication depends on the logic.
For instance, in classical propositional logic (over a finite signature), we have:
\begin{observationrep}
  \prop{Maximal Cut} implies \prop{Well-Founded} in propositional logic,
  but not in Horn logic.
\end{observationrep}
\begin{proof}
  Note first that in the original formulation by \citet{jandson:towards-contraction},
  \prop{Maximal Cut} is concerned with relations over CCTs,
  and the existence of maximal elements among all CCTs not containing a given~$\varphi$.
  \Citet{kr25:effectiveAGMContraction} show that in the case of LTL and BMA,
  this condition is equivalent to our formulation.
  Consider now classical propositional logic.
  Every set~$X$ of propositional CCTs can be described by a propositional formula~$\varphi$,
  so the properties are equivalent.
  For Horn, \cref{fig:horn-maxcut-wf} shows a relation violating \prop{Well-Founded},
  but satisfying \prop{Maximal Cut}.
  The set of CCTs on the cycle corresponds to the propositional formula~$\lnot p \land \lnot q$,
  which is however not part of the Horn fragment.
\end{proof}

\begin{toappendix}
\begin{figure}
  \centering
  \vspace*{-1em}
  \begin{tikzpicture}[thick,node distance=7mm]
    \node (empty) {$\Cn{\lnot p,\lnot q}$};
    \node[right=of empty] (p) {$\Cn{p,\lnot q}$};
    \node[below=of empty] (q) {$\Cn{\lnot p,q}$};
    \node[below=of p] (pq) {$\Cn{p,q}$};

    \draw[->] (q)  -- (empty);
    \draw[->] (p)  -- (empty);
    \draw[->] (pq) -- (empty);

    \draw[->,red] (p)  -- (pq);
    \draw[->,red] (pq) -- (q);
    \draw[->,red] (q)  -- (p);
  \end{tikzpicture}
  \captionsetup{hangindent=0pt,indention=0pt,labelfont=bf,format=plain}
  \caption{Mirroring relation that is not \prop{Well-Founded} (cycle in red), but satisfies \prop{Maximal Cut} wrt.\ Horn logic.}
  \label{fig:horn-maxcut-wf}
\end{figure}
\end{toappendix}

Generally, whether or not \prop{Maximal Cut} and \prop{Well-Founded} coincide
depends on two parameters.
The first parameter is the expressivity of single formulae in the logic:
as we have seen, limiting propositional logic to Horn prevents certain sets from being expressed as formulae,
leading to a rupture between the two properties.
In non-finitary logics such as LTL, this rupture is in fact unavoidable:
the countably many LTL formulae cannot capture the uncountably many sets of traces~\citep{kr25:effectiveAGMContraction}.
\begin{proposition}
  In LTL, \prop{Maximal Cut} and \prop{Well-Founded} do \emph{not} coincide in general.
\end{proposition}
\begin{proof}
  \Cref{prop:bma-not-wf-but-maxcut} below proves a stronger result.
\end{proof}
The second parameter is the class of relations we consider.
The observations above refer to the class of all preference relations.
Yet, in non-finitary logics like LTL, most relations between the infinitely many traces cannot be finitely represented,
and are thus not relevant when it comes to \emph{effective} contraction.
In the particular case of LTL,
we therefore ask ourselves whether the rupture between the two properties persists
when we restrict to (finitely-representable) BMA relations.
It turns out that this is the case:
\begin{propositionrep}
  \label{prop:bma-not-wf-but-maxcut}
  There exist BMA relations that
  satisfy \prop{Mirroring} and \prop{Maximal Cut} but do not satisfy \prop{Well-Founded}.
\end{propositionrep}
\begin{proofsketch}
  For every~$i\in\N$,
  let $\sigma_i := (\{p\}\{p\})^i\,\emptyset\,\{p\}^\omega$.
  Let  $\tau \pref \tau'$ hold if and only if
    (1)~$\tau = \sigma_i$ and $\tau' = \sigma_j$ for~$i,j\in\N$ with $i<j$; or
    (2)~$\tau = \sigma_i$ for~$i\in\N$, but $\tau' \neq \sigma_j$ for all~$j\in\N$.
  This relation does not satisfy \prop{Well-Founded}, as the~$\sigma_i$ form an infinite ascending sequence.
  But ${\pref}$~satisfies \prop{Mirroring} and \prop{Maximal Cut}.
  For \prop{Maximal Cut},
  we appeal to a result by \citet{wolper:more-expressive} showing that for every LTL formula~$\varphi$,
  either some trace~$\tau$ with $\tau\neq\sigma_j$ for all~$j\in\N$ satisfies~$\lnot\varphi$
  (and hence, $\tau$ is maximal in~$\tracesof{\lnot\varphi}$),
  or there exists a (finite) maximal~$i$ such that $\sigma_i$ satisfies~$\lnot\varphi$
  (and hence, $\sigma_i$ is maximal in~$\tracesof{\lnot\varphi}$).
  It remains only to note that there exists a BMA recognizing the relation~${\pref}$.
\end{proofsketch}
\begin{proof}
  Let $\AP=\{p\}$.
  For every~$i\in\N$,
  let $\sigma_i$~be the trace~$(\{p\}\{p\})^i\,\emptyset\,\{p\}^\omega$.
  We define a preference relation~$\pref$ such that  $\tau \pref \tau'$ holds if and only if
  \begin{itemize}
    \item $\tau = \sigma_i$ and $\tau' = \sigma_j$ for~$i,j\in\N$ with $i<j$; or
    \item $\tau = \sigma_i$ for~$i\in\N$, but $\tau' \neq \sigma_j$ for all~$j\in\N$.
  \end{itemize}
  This relation satisfies \prop{Mirroring}:
  If $\tau$ and $\tau'$ are incomparable, then $\tau,\tau'$ are both \emph{not} of the form~$\sigma_n$ for any $n$.
  But then there does not exist any trace~$\tau''$ such that $\tau \pref \tau''$, and therefore \prop{Mirroring} is trivially satisfied.

  At the same time, this relation does not satisfy \prop{Well-Founded}:
  The nonempty set $\{\,\sigma_i \mid i\in\N \,\}$ does not have a maximal element,
  as we have $\sigma_0 \pref \sigma_1 \pref \cdots$.

  However, the relation satisfies \prop{Maximal Cut}.
  Let $\varphi$ be a non-tautological LTL formula.
  We make a distinction between the following cases:
  \begin{description}
    \item[Case 1:] $\tracesof{\lnot\varphi}$ contains some~$\tau$ which is different from all~$\sigma_i$.\smallskip

      Then the trace~$\tau$ is maximal in~$\tracesof{\lnot\varphi}$.
    \item[Case 2:] There exist only finitely many traces $\tracesof{\lnot\varphi}$, and each of them is some $\sigma_i$.\smallskip

      In this case, there exists a largest~$i_{\max}\in\N$ such that $\sigma_{i_{\max}} \in \tracesof{\lnot\varphi}$, and thus the trace~$\sigma_{i_{\max}}$ is a maximal element of~$\tracesof{\lnot\varphi}$.
	\item[Case 3:] $\tracesof{\lnot\varphi}$ contains (only) infinitely many~$\sigma_i$.\smallskip

	  We show that this case is not possible.
      Suppose this were the case, and let $n$ be the number of ${\Next}$~operators in $\lnot\varphi$.
      Since there are infinitely many traces in $\tracesof{\lnot\varphi}$,
      there exists $m\in\N$ such that $2m > n$ and $\sigma_m= (\{p\}\{p\})^{m}\emptyset\{p\}^\omega \in \tracesof{\lnot\varphi}$.
      By a result of \citet{wolper:more-expressive},
      this implies that $(\{p\}\{p\})^m{\color{magenta}\{p\}}\emptyset\{p\}^\omega \in \tracesof{\lnot\varphi}$.%
      \footnote{
Intuitively, $\varphi$ cannot \emph{``see''} further than $n$ steps,
and can thus not distinguish between these two traces.
		}
      This however contradicts the assumption that $\tracesof{\lnot\varphi}$ contains only traces of the form~$\sigma_i$.
    \end{description}

  Finally, it remains to show that the relation~${\pref}$ is recognized by a BMA.
  To this end, consider the B\"uchi automaton~$A$ shown in \cref{fig:nba-sigmas},
  which accepts exactly all the traces~$\sigma_i$,
  and the BMA~$B_\sigma$ in \cref{fig:bma-sigmas},
  which recognizes the relation~$\{\,(\sigma_i,\sigma_j)\mid i<j\,\}$.
  We observe that we have ${\pref} = \rel{B_\sigma} \cup \big(\lang{A} \times \overline{\lang{A}}\big)$, by definition of~${\pref}$.
  Using the fact that B\"uchi automata are closed under complementation~\cite{buchi:s1s},
  \cref{lem:bma-closed-ops} and \cref{lem:buchi-product-bma} (shown further below),
  we conclude that there exists a BMA recognizing~${\pref}$.
\end{proof}
\begin{toappendix}
  \begin{figure}
    \begin{subfigure}{0.4\textwidth}
      \centering
      \begin{tikzpicture}[thick,inner sep=1,font=\scriptsize,node distance=15mm]
        \node[draw,circle] (q0) {$q_0$};
        \node[draw,circle,below of=q0] (q1) {$q_1$};
        \node[draw,double,circle,right of=q0] (q2) {$q_2$};
        \draw[<-] (q0) -- (left:7mm);
        \draw[->] (q0) edge[bend left] node[auto]{$\{p\}$} (q1);
        \draw[->] (q1) edge[bend left] node[auto]{$\{p\}$} (q0);
        \draw[->] (q0) -- node[auto]{$\emptyset$} (q2);
        \draw[->] (q2) edge[loop below] node[auto]{$\{p\}$} ();
      \end{tikzpicture}
      \caption{B\"uchi automaton~$A$ recognizing $\{\,\sigma_i \mid i\in\N\,\}$.}
      \label{fig:nba-sigmas}
    \end{subfigure}
    \hfill
    \begin{subfigure}{0.56\textwidth}
      \centering
      \begin{tikzpicture}[thick,inner sep=1,font=\scriptsize,node distance=15mm]
        \node[draw,circle] (q00) {$q_{00}$};
        \node[draw,circle,below of=q00] (q11) {$q_{11}$};
        \node[draw,circle,right=25mm of q00] (q21) {$q_{21}$};
        \node[draw,circle,below of=q21] (q20) {$q_{20}$};
        \node[draw,double,circle,right=25mm of q20] (q22) {$q_{22}$};
        \draw[<-] (q00) -- (left:7mm);
        \draw[->] (q00) edge[bend left] node[auto]{$\{p\}/\{p\}$} (q11);
        \draw[->] (q11) edge[bend left] node[auto]{$\{p\}/\{p\}$} (q00);
        \draw[->] (q00) -- node[auto]{$\emptyset/\{p\}$} (q21);
        \draw[->] (q21) edge[bend left] node[auto]{$\{p\}/\{p\}$} (q20);
        \draw[->] (q20) edge[bend left] node[auto]{$\{p\}/\{p\}$} (q21);
        \draw[->] (q20) -- node[auto]{$\{p\}/\emptyset$} (q22);
        \draw[->] (q22) edge[loop above] node[auto]{$\{p\}/\{p\}$} ();
      \end{tikzpicture}
      \caption{BMA~$B_\sigma$ recognizing $\{\,(\sigma_i,\sigma_j) \mid i,j\in\N \text{ with } i < j \,\}$.}
      \label{fig:bma-sigmas}
    \end{subfigure}
    \caption{Automata used in the proof of \cref{prop:bma-not-wf-but-maxcut}.}
  \end{figure}
\end{toappendix}
As before, the rupture stems from a limitation of the expressivity of single LTL formulae.
In this case, it is particularly the disconnect between the expressive powers of LTL resp.\ B\"uchi automata.

As we turn below to the decidability question for \prop{Well-Founded} resp.\ \prop{Maximal Cut} on BMA relations,
a restricted class of BMA relations in which the two properties coincide plays a key role.

\subsection{Undecidability of Maximal Cut}
\label{sec:maximal-cut}
\begin{toappendix}
  \subsection{Proofs for \cref{sec:maximal-cut}}
\end{toappendix}

We show that \prop{Maximal Cut} is undecidable for BMA relations via a reduction from a known undecidable problem.
Specifically, we reduce the \emph{immortality} of Turing machines (explained further below) to \emph{well-foundedness} of BMA relations.
\goodbreak
Central to our proof is the construction of a Büchi-Mealy automaton
that simulates a single step of a given Turing machine.
To extend the result to \prop{Maximal Cut},
we formulate our construction using \emph{linear bounded automata} (LBA), a special type of Turing machines.

Linear bounded automata function like regular Turing machines,
consisting of a \emph{read-write head} and a \emph{tape},
a string of symbols.
In each step, the head reads a symbol~$a$ on the tape,
and depending on the state~$q$ of the automaton, replaces~$a$ with a new symbol~$b$,
moves left or right,
and transitions into a new state~$q'$.
This behaviour is specified by a transition function~$\delta$.
What distinguishes LBA from general Turing machines is that they may only use the space on the tape that is initially occupied by the input.
\begin{definition}
  \label{def:lba}
    A \emph{linear bounded automaton}
    $T=(Q,\Sigma, \Gamma,q_0,\delta)$
    consists of
    a finite set of states~$Q$,
    a finite non-empty input alphabet~$\Sigma$,
    a finite tape alphabet~$\Gamma$ which includes $\Sigma$ and additional symbols
    $\leftend$ and $\rightend$ marking the left and right end of the tape,
    a start state $q_0 \in Q$,
    and a partial transition function $\delta: Q\times\Gamma\to Q\times\Gamma\times\{L,R\}$.
    The function $\delta$ must not include transitions that overwrite $\leftend$ or $\rightend$ with other symbols,
    move to the left of~$\leftend$ resp.\ the right of~$\rightend$,
    or overwrite other symbols with $\leftend$ or $\rightend$.
\end{definition}
Every combination of state $q\in Q$, tape content $w\in\Gamma^*$ and head position, denoted by the $\bullet$-symbol above the corresponding symbol in $w$,
form a \emph{configuration} of a given linear bounded automaton.
$qu\headp{a}v$ denotes the configuration where $w=uav$ for $u,v\in \Gamma^*$ and the automaton is in state $q$ with the head above the symbol $a$.
The function~$\delta$ defines a successor relation $\to$ on configurations of $T$ in the expected way.
For example: $q\leftend\headp{a}\rightend \to q'\leftend b\headp{\rightend} $ if $\delta(q,a)=(q',b,R)$.
A configuration~$c$ is \emph{halting} if it has no successor,
\emph{mortal} if there is a halting configuration $c'$ such that $c\to^* c'$,
and \emph{immortal} otherwise.

The following well-known property of linear bounded automata is central to our proof.
\begin{toappendix}
  We repeat the proof of \cref{prop:lba-infiniterun} for the sake of self-containedness.
\end{toappendix}
\begin{lemmarep}[\citet{chowdhary:lba}]
    \label{prop:lba-infiniterun}
    Any infinite run of successive configurations $c_1\to c_2\to\dots$ of an LBA~$T$ only consists of finitely many unique configurations.
\end{lemmarep}
\begin{proofsketch}
    The tape never changes length, so all successive configurations must have the same tape length~$n$.
    There are only finitely many configurations of a given length.
\end{proofsketch}
\begin{proof}
    Let $c$ and $c'$ be configuration of $T$. In a single step, $T$ only changes one symbol on the tape, no symbols can be added or removed from the tape, so if $c\to c'$, then $c'$ has to have the same length as $c$ and the same holds for all configurations of the run $\{c_i\mid i\in \N \}$. Let $n$ denote the length of the tape in $c_1$ (excluding $\leftend$ and $\rightend$). All configurations $c_i$ are among the $|Q|(n+2)|\Gamma - 2|^n$ distinct configurations of the same length: $|Q|$ is the number of possible states, $n+2$ the number of possible head positions and $|\Gamma - 2|^n$ the number of distinct tape contents of length $n$.
\end{proof}

For our proof, we use a variant of the halting problem, which is undecidable for Turing machines as well as linear bounded automata.
This \emph{immortality problem}~\citep{hooper:immortality} asks whether or not a given machine has any immortal configuration.
No input is specified, in fact the immortal configuration does not even have to be reachable from any initial configuration.
\begin{proposition}[\citet{hooper:immortality}]
	\label{thm:lba-immortality}
    The \emph{LBA Immortality Problem}, as given below, is undecidable.
  \decproblem{A linear bounded automaton $T$}{~
  $T$ has an immortal configuration}
\end{proposition}
Our reduction proof rests on the following construction.
\begin{toappendix}
    Before we go on to the details of our construction, we quickly have to address the technical details of how we encode configurations of $T$ for our Büchi-Mealy automaton $B_T$.
    \begin{definition}
        Let $T=(Q,\Sigma,\Gamma,q_0,\delta)$ be a linear bounded automaton. $B_T$ must be defined over a set of atomic propositions~$\AP$, so all symbols needed to describe configurations of $T$ need to be represented with subsets of $\AP$. Therefore, we must choose $\AP$ to be big enough to allow for such a representation, formally an injective function
        \[\rho: Q \cup \Gamma \cup \headp{\Gamma} \to \powerset{\AP},  \]
        where $\headp{\Gamma}:=\{\headp{a}\mid a\in \Gamma\}$ is a new alphabet, which is needed to mark the head position of $T$.
        For notational clarity, we denote all sets $A\in\powerset{\AP}$ by their preimage $\rho^{-1}(A)$, e.g. we will say the BMA reads a symbol $a\in \Gamma$ instead of some set of atomic propositions $\rho(a)\in\powerset{\AP}$, which represents $a$.
        Let $C_T$ denote the set of all configurations of $T$, the encoding of a configuration is then given by a function
        \[\enc: C_T \to \powerset{\AP}^*. \]
        Let $c$ be a configuration, all encodings are in one of the following forms:
        \begin{enumerate}
            \item $\enc(c)=q\headp{\leftend}w\rightend$,
            \item $\enc(c)=q\leftend u\headp{a}v\rightend$,
            \item $\enc(c)=q\leftend w\headp{\rightend}$,
        \end{enumerate}
        where $q\in Q;a\in \Iota ;u,v,w\in \Iota^*$ and $\Iota:=\Gamma\setminus\{\leftend,\rightend\}$ denotes the tape alphabet excluding the end-of-tape symbols.
    \end{definition}
    \begin{definition}
      \label{def:lba-trans-bma}
        Let~$T$ be a linear bounded automaton, $q\in Q$ a state and $a\in \Gamma$ a symbol from the tape alphabet.
        The B\"uchi-Mealy automaton~$B_{q,a}$ is defined as follows:
        \begin{enumerate}
            \item If there is no transition $\delta(q,a)$, we define $B_{q,a}$ as the following BMA (for the empty language):
            \begin{center}
                \begin{tikzpicture}[->,>=stealth',shorten           >=1pt,auto,node distance=2cm, semithick,
                    everystate/.style={fill=white,draw=black,text=black},
                        ]
    
                \node[state,inner sep=1pt,minimum size=0pt,initial, initial text=] (q0) {$q_0$};
                \end{tikzpicture}
            \end{center}
            \item If $\delta(q,a)=(q',a',R)$ and $a\in\Iota$, we define $B_{q,a}$ as
            \begin{center}\scalebox{.8}{
                \begin{tikzpicture}[->,>=stealth',shorten           >=1pt,auto,node distance=2cm, semithick,
                    everystate/.style={fill=white,draw=black,text=black},
                                ]
                
                    \node[state,inner sep=1pt,minimum size=0pt,initial, initial text=] (q0){$q_0$};
                    \node[state,inner sep=1pt,minimum size=0pt] (q1) [right of=q0]       {$q_1$};
                    \node[state,inner sep=1pt,minimum size=0pt] (q2) [right of=q1]       {$q_2$};
                    \node[state,inner sep=1pt,minimum size=0pt] (q3) [right of=q2]       {$q_3$};
                    \node[state,inner sep=1pt,minimum size=0pt] (q4) [right of=q3]       {$q_4$};
                    \node[state,inner sep=1pt,minimum size=0pt,accepting] (q5) [right of=q4]       {$q_5$};
                
                    \path (q0) edge[above] node{$q/q'$} (q1)
                    (q1) edge[above] node{$\leftend/\leftend$} (q2)
                    (q2) edge[loop above] node{$\{b/b\}_{b\in \text{I}}$} (q2)
                    (q2) edge[above] node{$\headp{a}/a'$} (q3)
                    (q3) edge[above] node{$\{b/\headp{b}\}_{b\in \text{I}}$} (q4)
                    (q3) edge[below, bend right] node{$\rightend/\headp{\rightend}$} (q5)
                    (q4) edge[loop above] node{$\{b/b\}_{b\in \text{I}}$} (q4)
                    (q4) edge[above] node{$\rightend/\rightend$} (q5)
                    (q5) edge[loop below] node{$\Sigma/\Sigma$} (q5)
                    ;
                \end{tikzpicture}}
            \end{center}
            where $\{b/b\}_{b\in \text{I}}$ and $\{b/\headp{b}\}_{b\in \text{I}}$ denote the set of all symbol-pairs of the given form. As a transition label such a set is a short-hand for all transitions with labels from the set. Note that the transitions here need to match the pattern, while $\Sigma/\Sigma$ denotes all possible transitions: The automaton will always transition to the next state, regardless of which symbol-pair is read.
            \item If $\delta(q,a)=(q',a',L)$ and $a\in\Iota$, we define $B_{q,a}$ as
            \begin{center}\scalebox{.8}{
                \begin{tikzpicture}[->,>=stealth',shorten           >=1pt,auto,node distance=2cm, semithick,
                    everystate/.style={fill=white,draw=black,text=black},
                                ]
                
                    \node[state,inner sep=1pt,minimum size=0pt,initial, initial text=] (q0)                    {$q_0$};
                    \node[state,inner sep=1pt,minimum size=0pt] (q1) [right of=q0]       {$q_1$};
                    \node[state,inner sep=1pt,minimum size=0pt] (q2) [right of=q1]       {$q_2$};
                    \node[state,inner sep=1pt,minimum size=0pt] (q3) [right of=q2]       {$q_3$};
                    \node[state,inner sep=1pt,minimum size=0pt] (q4) [right of=q3]       {$q_4$};
                    \node[state,inner sep=1pt,minimum size=0pt,accepting] (q5) [right of=q4]       {$q_5$};
                
                    \path (q0) edge[above] node{$q/q'$} (q1)
                    (q1) edge[above] node{$\leftend/\leftend$} (q2)
                    (q2) edge[loop above] node{$\{b/b\}_{b\in \text{I}}$} (q2)
                    (q3) edge[above] node{$\headp{a}/a'$} (q4)
                    (q2) edge[above] node{$\{b/\headp{b}\}_{b\in \text{I}}$} (q3)
                    (q1) edge[below, bend right] node{$\leftend/\headp{\leftend}$} (q3)
                    (q4) edge[loop above] node{$\{b/b\}_{b\in \text{I}}$} (q4)
                    (q4) edge[above] node{$\rightend/\rightend$} (q5)
                    (q5) edge[loop below] node{$\Sigma/\Sigma$} (q5)
                    ;
                \end{tikzpicture}}
            \end{center}
            \item If $\delta(q,a)=(q',a',D)$ and $a=\leftend$, then by \cref{def:lba}, we must have $D=R$ and $a'=\leftend$.
              We define $B_{q,a}$ as
            \begin{center}\scalebox{.8}{
                \begin{tikzpicture}[->,>=stealth',shorten           >=1pt,auto,node distance=2cm, semithick,
                    everystate/.style={fill=white,draw=black,text=black},
                                ]
                
                    \node[state,inner sep=1pt,minimum size=0pt,initial, initial text=] (q0)                    {$q_0$};
                    \node[state,inner sep=1pt,minimum size=0pt] (q1) [right of=q0]       {$q_1$};
                    \node[state,inner sep=1pt,minimum size=0pt] (q2) [right of=q1]       {$q_2$};
                    \node[state,inner sep=1pt,minimum size=0pt] (q3) [right of=q2]       {$q_3$};
                    \node[state,inner sep=1pt,minimum size=0pt,accepting] (q4) [right of=q3]       {$q_4$};
                
                    \path (q0) edge[above] node{$q/q'$} (q1)
                    (q1) edge[above] node{$\headp{\leftend}/\leftend$} (q2)
                    (q3) edge[loop above] node{$\{b/b\}_{b\in \text{I}}$} (q3)
                    (q2) edge[above] node{$\{b/\headp{b}\}_{b\in \text{I}}$} (q3)
                    (q3) edge[above] node{$\rightend/\rightend$} (q4)
                    (q2) edge[below, bend right] node{$\rightend/\headp{\rightend}$} (q4)
                    (q4) edge[loop below] node{$\Sigma/\Sigma$} (q4)
                    ;
                \end{tikzpicture}}
            \end{center}
            \item If $\delta(q,a)=(q',a',D)$ and $a=\rightend$, then by \cref{def:lba}, we must have $D=L$ and $a'=\rightend$.
              We define $B_{q,a}$ as
            \begin{center}\scalebox{.8}{
                \begin{tikzpicture}[->,>=stealth',shorten           >=1pt,auto,node distance=2cm, semithick,
                    everystate/.style={fill=white,draw=black,text=black},
                                ]
                
                    \node[state,inner sep=1pt,minimum size=0pt,initial, initial text=] (q0)                    {$q_0$};
                    \node[state,inner sep=1pt,minimum size=0pt] (q1) [right of=q0]       {$q_1$};
                    \node[state,inner sep=1pt,minimum size=0pt] (q2) [right of=q1]       {$q_2$};
                    \node[state,inner sep=1pt,minimum size=0pt] (q3) [right of=q2]       {$q_3$};
                    \node[state,inner sep=1pt,minimum size=0pt,accepting] (q4) [right of=q3]       {$q_4$};
                
                    \path (q0) edge[above] node{$q/q'$} (q1)
                    (q1) edge[above] node{$\leftend/\leftend$} (q2)
                    (q1) edge[below, bend right] node{$\leftend/\headp{\leftend} $} (q3)
                    (q2) edge[loop above] node{$\{b/b\}_{b\in \text{I}}$} (q2)
                    (q2) edge[above] node{$\{b/\headp{b}\}_{b\in \text{I}}$} (q3)
                    (q3) edge[above] node{$\headp{\rightend}/\rightend$} (q4)
                    (q4) edge[loop below] node{$\Sigma/\Sigma$} (q4)
                    ;
                \end{tikzpicture}}
            \end{center}
        \end{enumerate}
    \end{definition}
    \begin{lemma}\label{lem:construction-lemma1}
        A pair of traces $(\tau_1,\tau_2)\in\Sigma_{\textnormal{BM}}^\omega$ is accepted by $B_{q,a}$ if and only if
        \begin{enumerate}[label={(\arabic*)}]
            \item they represent valid configurations i.e. there are configurations $c_1,c_2\in C_T $ and traces $\tau'_1,\tau'_2\in\powerset{\AP}^\omega$ such that
            \[\tau_1=\enc(c_1)\tau'_1\ \ \text{and}\ \  \tau_2=\enc(c_2)\tau'_2,\]
            \item those configurations form a succession $c_1\to c_2$ and
            \item $c_1$ matches $q,a$ i.e. in configuration $c_1$ the machine is in state $q$ and positioned above the symbol $a$.
        \end{enumerate}
    \end{lemma}
    \begin{proof}
        It is clear from the construction (\cref{def:lba-trans-bma}) that all accepted trace pairs are of the form described.
        And all traces that fit the description are accepted by~$B_{q,a}$ because the construction is complete, i.e., it covers all possible forms of transitions.
    \end{proof}
    \begin{lemma}\label{lem:construction-lemma2}
        For any pair of configurations~$c_1,c_2\in C_T$, it holds that
        \[c_1\to c_2\]
        if and only if there exist $q\in Q$ and $a\in \Gamma$ such that
        \[(\enc(c_1)\tau_1,\enc(c_2)\tau_2)\in\rel{B_{q,a}}\]
        for all $\tau_1,\tau_2\in\powerset{\AP}^\omega$.
    \end{lemma}
    \begin{proof}
        The configuration $c_1$ matches some combination $q,a$.
        By \cref{lem:construction-lemma1}, $c_1 \to c_2$ holds if and only if $\rel{B_{q,a}}$ satisfies the latter condition.
    \end{proof}
\end{toappendix}

\begin{lemmarep}\label{lemma:lba-to-bma}
    For a given LBA~$T$,
    we can construct a Büchi-Mealy automaton $B_T$ such that $\rel{B_T}$ is \prop{Well-Founded} if and only if $T$ has no immortal configuration.
\end{lemmarep}
\begin{proofsketch}
    We construct a Büchi-Mealy automaton that accepts only trace-pairs that describe successive configurations $c_1 \to c_2$ by non-deterministically guessing the state $q$ and the current symbol~$a$ in~$c_1$,
    and then checking if the two configurations match according to the rule~$\delta(q,a)$.
    For instance, if $\delta(q,a)=(q',a',R)$, the automaton expects states~$q$ and~$q'$ in the first and second trace respectively,
    then it expects the traces to match symbol-by-symbol until it finds~$\headp{a}$ and~$a'$ in the first and second trace respectively.
    I.e., the head was above symbol~$a$ on the tape, which is now replaced by~$a'$ and the head is positioned elsewhere.
    The head moves to the right according to the rule, so the automaton expects it on the next symbol in the second trace.
    Afterwards, the traces must match again until the symbol~$\rightend$ appears on both traces, marking the end of the configuration. 

    Every immortal configuration~$c$ corresponds to an infinite run $c\to c_1\to c_2\to\dots$ of~$T$ and thus to to an infinite ascending chain in $\rel{B_T}$, and vice versa.
\end{proofsketch}
\begin{proof}
    By \cref{{lem:bma-closed-ops}} we can construct $B_T$ such that  
    \[\rel{B_T} = \bigcup_{q\in Q}\bigcup_{a\in\Gamma} \rel{B_{q,a}}.\]
    If $\rel{B_T}$ is \prop{Well-Founded}, then all configurations of $T$ must be mortal. Otherwise there would be an infinite sequence of configurations $c_1\to c_2\to \cdots$.
    By \cref{lem:construction-lemma2}, this means there is an infinite ascending chain $\enc(c_1)\tau_1,\enc(c_2)\tau_2,\cdots$,
    where
    \[ (\enc(c_i)\tau_i,\enc(c_{i+1})\tau_{i+1})\in\rel{B_T} \]
    for all $i \in \N$.
    On the other hand, if $\rel{B_T}$ is not \prop{Well-Founded}, then by \cref{lem:construction-lemma1}, the infinite ascending chain is of the form $\enc(c_1)\tau_1,\enc(c_2)\tau_2,\cdots$ with $c_i\to c_{i+1}$ for all $i \in \N$.
    Hence, $c_1$~is immortal.
\end{proof}

This construction allows us to reduce LBA immortality to BMA well-foundedness.
\begin{theoremrep}
    The Well-Foundedness problem for Büchi-Mealy automata, as given below, is undecidable.
  \decproblem{A B\"uchi-Mealy automaton $B$}{$\rel{B}$ satisfies \prop{Well-Founded}}
\end{theoremrep}
\begin{proof}
  \cref{{lemma:lba-to-bma}} shows that the LBA immortality problem reduces to well-foundedness for Büchi-Mealy automata. The LBA immortality problem is undecidable, so well-foundedness for Büchi-Mealy automata must be undecidable as well.
\end{proof}

Let us shift our focus back to \prop{Maximal Cut},
our original property of interest.
We show its undecidability by linking it once again to well-foundedness
-- indeed, on relations corresponding to LBA, these two properties coincide:
\begin{lemmarep}
\label{lem:lba-wf-iff-maxcut}
    The relation~$\rel{B_T}$, for an LBA~$T$, is well-founded if and only if it satisfies \prop{Maximal Cut}.
\end{lemmarep}
\begin{proofsketch}
  We need only show that \prop{Maximal Cut} implies \prop{Well-Founded}; the opposite direction always holds.
  Therefore, suppose contrapositively that $\rel{B_T}$ were not well-founded.
  Then by \cref{lemma:lba-to-bma}, $T$~has some immortal configuration
  from which there exists an infinite run.
  By \cref{prop:lba-infiniterun}, this run consists of only finitely many configurations.
  We construct an LTL formula~$\varphi$ that describes exactly all the configurations of this run.
  Then, $\tracesof{\varphi}$ cannot have a maximal element (it would correspond to a halting configuration in the infinite run).
  Hence, $\rel{B_T}$ violates \prop{Maximal Cut}; the formula~$\lnot\varphi$ serves as the counter-example.
\end{proofsketch}
\begin{proof}
    We need only show that \prop{Maximal Cut} implies \prop{Well-Founded}; the opposite direction always holds.

    Contrapositively, let us assume that $\rel{B_T}$ is not \prop{Well-Founded}, then there exists an infinite ascending chain which by \cref{lem:construction-lemma1}
    is of the form $\enc(c_1)\tau_1,\enc(c_2)\tau_2,...$ with $c_i\to c_{i+1}$ for all $i \in \N$.  By \cref{prop:lba-infiniterun}, the infinite sequence $c_1\to c_2\to \dots$ consists of only finitely many unique configurations $c'_1,\dots,c'_n$. Let
    \[\varphi = \bigvee_{i=1}^n \bigwedge_{j=1}^m \Next^j \left(\bigwedge_{p\in \enc(c'_i)_j} p \land \bigwedge_{p\in\AP\setminus \enc(c'_i)_j} \lnot p  \right)  \]
    where $m=|\enc(c'_1)|$ is the length of one encoding (all are the same length). Note that $\enc(c'_i)$ is a word, $\enc(c'_i)_j$, its $j$-th symbol, a set of atomic propositions. To satisfy $\varphi$ a trace must match one of the encodings $\enc(c'_1),...,\enc(c'_n)$ symbol for symbol in the first $m$ positions, thus
    $\tracesof{\varphi}$ consists of exactly those traces $\tau\in\powerset{\AP}^\omega$ that are of the form $\tau=\enc(c'_i)\tau'$ for some $i\leq n$ and $\tau'\in\powerset{\AP}^\omega$. So these traces all correspond to the configurations $c'_1,...,c'_n$ and by \cref{lem:construction-lemma2} none of them can be maximal, that would mean their corresponding configuration is halting, which cannot be the case as these configurations come from an infinite sequence and are immortal. Therefore $\max_{\rel{B_T}}{\tracesof{\varphi}}$ is empty, and so $\rel{B_T}$ violates \prop{Maximal Cut}; the (non-tautological) formula~$\lnot\varphi$ serves as the counter-example.
\end{proof}
Hence, by the same construction as before, we can also reduce the LBA immortality problem to \prop{Maximal Cut}.
This is the last ingredient for our central theorem:
\begin{theoremrep}
  \label{thm:max-cut-undec}
   The following problem is undecidable:
  \decproblem{A B\"uchi-Mealy automaton $B$}{$\rel{B}$ satisfies \prop{Maximal Cut}}
\end{theoremrep}
\begin{proof}
  By \cref{{lemma:lba-to-bma}} and \cref{lem:lba-wf-iff-maxcut} LBA immortality reduces to \prop{Maximal Cut} for Büchi-Mealy automata. The LBA immortality problem is undecidable, so \prop{Maximal Cut} for Büchi-Mealy automata must be undecidable as well.
\end{proof}
Interestingly, our proof of undecidability does not even require the full expressive power of LTL.
In fact, the only temporal operator used to construct the formula~$\varphi$ in \cref{lem:lba-wf-iff-maxcut} is ${\Next\!}$ (``next'').
The fixpoint operators ($\Until$, or the derived operators~$\Finally$ resp.\ $\Globally$) that allow reasoning over an unbounded temporal horizon are not needed for this result.

We conclude this section with a discussion on the impact of our undecidability result, namely that it implies the impossibility of a \emph{constructive representation result}. To that end, we note:
\begin{observationrep}
  The \prop{Maximal Cut} problem on BMA, as given in \cref{thm:max-cut-undec}, is co-semi-decidable.
\end{observationrep}
\begin{proof}
  To find a violation of \prop{Maximal Cut} on a given BMA~$B$,
  successively enumerate all non-tautological LTL formulae.
  For each such formula~$\varphi$, construct a B\"uchi automaton recognizing ${\max}_{\rel{B}} \tracesof{\lnot\varphi}$ (as detailed in~\citep{kr25:effectiveAGMContraction})
  and determine whether it accepts the empty language.
  If so, then $B$~does not satisfy \prop{Maximal Cut}.
\end{proof}
The source of undecidability for \prop{Maximal Cut}
are thus the cases where the property indeed holds. It is not semi-decidable, i.e.
it is impossible to \emph{certify} the property in finite time.
Suppose now a (finite) system of effective constructions of BMA,
such that the resulting BMA are guaranteed to satisfy \prop{Maximal Cut} by design.
Indeed, we give just such a system in the next section.
Our observation above implies that such a system \emph{cannot possibly be complete},
i.e., there can be no representation theorem for \prop{Maximal Cut}.
There always must exist some BMA relations satisfying \prop{Maximal Cut} that cannot be constructed with this system.
Contradicting our observation, such a complete system would provide a semi-decision procedure for \prop{Maximal Cut}:
Given a BMA~$B$, enumerate all BMA~$B'$ constructible with the system and decide one-by-one whether $B$~and~$B'$ are equivalent.

Similarly, this observation precludes the possibility of a stronger, decidable (perhaps even syntactic) condition \textbf{(MC+)} on BMA
that implies \prop{Maximal Cut}
and is ``quasi-complete'',
in the sense that
every BMA satisfying \prop{Maximal Cut} is \emph{equivalent} to one that satisfies \textbf{(MC+)}.
This also yields a semi-decision procedure: to verify \prop{Maximal Cut} for a given~$B$, enumerate all BMA~$B'$,
and decide for each of them whether $B'$~satisfies \textbf{(MC+)},
and if so, whether $B$~and~$B'$ are equivalent.

Although we cannot overcome undecidability with easy alternatives, the goal of effective LTL contraction is not invalidated.
Rather this result motivates us to study epistemic preference relations suited for particular applications of contraction.
The next section presents a first step in this direction.

\section{Constructions for Rational Preferences}
\label{sec:constructions}
As maximal cut is undecidable,
we cannot in general take a proposed BMA~$B$ and automatically check whether its relation satisfies \prop{Maximal Cut},
in addition to~\prop{Mirroring},
thus certifying that $B$~may be used to perform fully rational contractions.
Instead, we propose in the following a number of constructions
that guarantee \emph{by design} that the resulting BMA relations satisfy both conditions.

As a first step,
we analyze more closely the class of relations that satisfy \prop{Mirroring} and \prop{Maximal Cut}.
To this end, we introduce a subclass of preference relations, called \emph{rankings}.
\begin{definition}
  A relation~${\pref}\subseteq \Sigma^\omega\times \Sigma^\omega$ is a \emph{ranking} over a totally-ordered set~$(Y,{<_Y})$
  if there exists a function~$f:\Sigma^\omega\to Y$
  such that $\tau_1 \pref \tau_2$ if and only if $f(\tau_1) <_Y f(\tau_2)$.
\end{definition}
\begin{toappendix}
  
\subsection{Proof of \cref{prop:mirroring-maxcut-ranking,thm:repr-result}}
We begin this appendix section by proving the connection between the properties \prop{Mirroring} and \prop{Maximal Cut} on one side,
and rankings on the other side~(\cref{prop:mirroring-maxcut-ranking}).
To this end, we establish several lemmata.

We use the notion of \emph{ultimately periodic} (UP) traces,
i.e., traces of the form~$\tau = \tau_\textrm{stem}\,\tau_\textrm{loop}^\omega$ with a finite prefix~$\tau_\textrm{stem}$ followed by an infinitely-repeated finite sequence~$\tau_\textrm{loop}$.
It has been shown by \citet{kr25:effectiveAGMContraction} that for every UP trace~$\tau$,
there exists an \emph{identifying formula}~$\id{\tau}$ such that $\tau$~is the only trace that satisfies~$\id{\tau}$.
We use this insight below.

First, we observe that a key property of rankings is that they are \emph{acyclic}.
In particular, for BMA relations, the existence of any cycle implies the existence of a cycle over only UP traces:
\begin{lemma}
  \label{lem:bma-acyclic-UP}
  Let $B$ be a BMA.
  The relation~$\rel{B}$ is acyclic
  if and only if
  the restriction $\rel{B} \cap \UP^2$ to UP traces is acyclic.
\end{lemma}
\begin{proof}
  The forwards direction is trivial: a cycle of UP traces is a cycle of traces.

  We show the backwards direction here for cycles of size 2, but the proof is analogous for arbitrary cycles.
  Hence, assume contrapositively that a cycle of (arbitrary) traces $\tau_1,\tau_2, \tau_1$ existed,
  i.e., assume $(\tau_1,\tau_2)\in\rel{B}$ and $(\tau_2,\tau_1)\in\rel{B}$.
  Then there must exist an accepting run $\rho_{12} = q_0q_1q_2\ldots$ for $(\tau_1,\tau_2)$
  and an accepting run $\rho_{21} = q'_0 q'_1 q'_2 \ldots$ for $(\tau_2, \tau_1)$.

  As $\rho_{12}$ is accepting, some accepting state $q$ must occur infinitely often on it.
  Now, consider the infinite set $\{\, i\in\N \mid q_i = q \,\}$, and consider the $q'_i$ for such~$i$.
  As there are only finitely many states of $B$, some state $q'$ must occur infinitely often among the $q'_i$.

  Let $k$ be the minimal such index with $q_k = q$ and $q'_k = q'$.
  As there are infinitely many such indices, there is also some $l>k$ with $q_l = q$ and $q'_l = q'$.
  In particular, we choose $l$ large enough such that between $q'_k$ and $q'_{k+l}$,
  there is at least one accepting state on $\rho_{21}$
  (such an $l$ exists, because $\rho_{21}$ contains infinitely many occurrences of accepting states).
  Then we define the UP traces
  \begin{align*}
    \tau_1' &:= \tau_1[1..k]\,\tau_1[k+1..k+l]^\omega\\
    \tau_2' &:= \tau_2[1..k]\,\tau_2[k+1..k+l]^\omega
  \end{align*}
  The run $\rho_{12}' := \rho_{12}[0..k]\,\rho_{12}[k+1..k+l]^\omega$
  is an accepting run of~$B$ for $(\tau_1',\tau_2')$.
  And similarly, the run $\rho_{21}' := \rho_{21}[0..k]\,\rho_{21}[k+1..k+l]^\omega$ is an accepting run of~$B$ for $(\tau_2', \tau_1')$.
  Therefore, we have $(\tau_1',\tau_2')\in\rel{B}$ and $(\tau_2',\tau_1')\in\rel{B}$,
  and hence, a cycle in $\rel{B} \cap \UP^2$.
\end{proof}
We use the above lemma in order to show:
\begin{lemma}
  \label{lem:bma-maxcut-acyclic}
  Given a BMA~$B$,
  if $\rel{B}$~satisfies \prop{Maximal Cut},
  then $\rel{B}$~is acyclic.
\end{lemma}
\begin{proof}
  For purposes of contraposition, suppose there exists a cycle $\tau_1, \tau_2, \ldots \tau_n, \tau_1$, with $\tau_1,\ldots,\tau_n\in\Sigma^\omega$.
  We show that then $\rel{B}$~does not satisfy \prop{Maximal Cut}.

  By \cref{lem:bma-acyclic-UP}, we can assume wlog.\ that indeed $\tau_1,\ldots,\tau_n\in\UP$ holds.
  We consider the formula $\varphi :\equiv \lnot (\id{\tau_1} \lor \ldots \lor \id{\tau_n})$.
  The contra-traces of~$\varphi$ are $\tracesof{\lnot\varphi} = \tracesof{\id{\tau_1} \lor \ldots \lor \id{\tau_n}} = \{\tau_1,\ldots,\tau_n\}$.
  As these traces form a cycle in~$\rel{B}$, we have that $\max_{\rel{B}} \tracesof{\lnot\varphi} = \emptyset$.
  This violates \prop{Maximal Cut}.
\end{proof}
In the following, we use the notation $x_1 \incomp[R] x_2$ for the \emph{incomparability relation} derived from a relation~$R$,
i.e., $x_1 \incomp[R] x_2$ holds if and only if $(x_1,x_2)\notin R$ and $(x_2,x_1)\notin R$.

We use this notation for instance to prove the following lemma, which shows that combining acyclicity with \prop{Mirroring} yields transitivity:
\begin{lemma}
  \label{lem:acyclic-mirroring-transitive}
  Let $R \subseteq X \times X$ be an acyclic relation (over some set $X$) that satisfies \prop{Mirroring}.
  The relation~$R$ is transitive.%
\end{lemma}
\begin{proof}
  Let $(x_1, x_2) \in R$ and $(x_2, x_3)\in R$.
  Then we must either have $(x_1, x_3)\in R$, $(x_3, x_1)\in R$, or $x_1 \incomp[R] x_3$.
  In the first case, we are done.
  We show that the second and third case are impossible, as they lead to contradictions:
  \begin{description}
    \item[case $(x_3,x_1)\in R$:]
      Then $x_1,x_2,x_3,x_1$ is a cycle,
      contradicting the assumption that $R$ is acyclic.
    \item[case {$x_1 \incomp[R] x_3$}:]
      By \prop{Mirroring} and $(x_1,x_2)\in R$, we have $(x_3,x_2)\in R$.
      Then $x_2,x_3,x_2$ is a cycle,
      again contradicting the assumption that $R$ is acyclic.
  \end{description}
  Hence, transitivity holds.
\end{proof}
In addition, we show that the incomparability relation is an equivalence:
\begin{lemma}
  \label{lem:irreflexive-mirroring-equivalence-incomp}
  Let $R \subseteq X \times X$ be a irreflexive relation (over some set $X$) that satisfies \prop{Mirroring}.
  The incomparability relation $\incomp[R]$ is an equivalence relation.
\end{lemma}
\begin{proof}
  We consider the individual properties:
  \begin{description}
    \item[reflexivity:] Follows directly from irreflexivity of $R$.
    \item[symmetry:] By definition.
    \item[transitivity:] Let $x_1 \incomp[R] x_2$ and $x_2 \incomp[R] x_3$.
      We must have one of the following three cases:
      \begin{description}
        \item[case $(x_1,x_3)\in R$:] By \prop{Mirroring} and $x_1 \incomp[R] x_2$,
          we have $(x_2, x_3)\in R$,
          contradicting the assumption that $x_2 \incomp[R] x_3$.
        \item[case $(x_3,x_1)\in R$:] By \prop{Mirroring} and $x_2 \incomp[R] x_3$,
          we have $(x_2, x_1) \in R$,
          contradicting the assumption that $x_1 \incomp[R] x_2$.
        \item[case {$x_1 \incomp[R] x_3$}:] This is the only possible case,
          which is precisely what we needed to show.
      \end{description}
  \end{description}
  Hence, $\incomp[R]$ is an equivalence relation.
\end{proof}
The lemma above also applies to mirroring, acyclic relations, as they are in particular irreflexive.
In the next step, we show that an acyclic, mirroring relation~$R$ (over some set~$X$)
can in fact be characterized by a total strict order~$S$ on a partition of~$X$.
To this end, let $[x]$~denote the ${\incomp[R]}$-equivalence class of an element~$x\in X$,
and let $S$~be the relation over these equivalence classes such that
\[
  ([x_1],[x_2])\in S \iff (x_1, x_2)\in R
\]
We first observe:
\begin{lemma}
  \label{lem:eqclass-order-defined}
  The relation $S$ is well-defined.
\end{lemma}
\begin{proof}
  Let $[x_1] = [x_1']$ (i.e., $x_1 \incomp[R] x_1'$) and $[x_2] = [x_2']$ (i.e., $x_2 \incomp[R] x_2'$).
  We have to show that $(x_1,x_2)\in R$ if and only if $(x_1', x_2')\in R$.
  We show the forwards direction; the backwards direction is symmetric.
  
  Hence, let $(x_1,x_2)\in R$.
  By \prop{Mirroring} and $x_1 \incomp[R] x_1'$, we have $(x_1',x_2)\in R$.
  We distinguish three cases:
  \begin{description}
    \item[case $(x_2', x_1')\in R$:]
      By \prop{Mirroring} and $x_2' \incomp[R] x_2$, we have $(x_2,x_1')\in R$.
      This results in a cycle $x_1',x_2,x_1'$, contradicting the fact that $R$ is acyclic.
    \item[case {$x_1' \incomp[R] x_2'$}:]
      Transitivity of ${\incomp[R]}$ implies that $x_1' \incomp[R] x_2$,
      which contradicts the fact that $(x_1',x_2)\in R$.
    \item[case $(x_1',x_2')\in R$:] This is the only possible case,
      which is precisely what we needed to show.
      \qedhere
  \end{description}
\end{proof}
And moreover, the relation~$S$ is a total strict order:
\begin{lemma}
  \label{lem:eqclass-order-strictorder}
  The relation~$S$ is a total strict order,
  i.e., total, acyclic and transitive.
\end{lemma}
\begin{proof}
  We show each property individually:
  \begin{description}
    \item[totality:] Let $[x_1] \neq [x_2]$ be two equivalence classes of ${\incomp[R]}$.
      As they are distinct, $x_1$ and $x_2$ are comparable,
      i.e., either $(x_1,x_2)\in R$ or $(x_2,x_1)\in R$.
      Hence, we have $([x_1],[x_2])\in S$ or $([x_2],[x_1])\in S$.
    \item[acyclicity:]
      Suppose there existed a cycle $[x_1],\ldots,[x_n]$ with $([x_1],[x_2])\in S$, \ldots, $([x_{n-1},x_n])\in S$ and $([x_n],[x_1])\in S$.
      Then by definition of $S$,
      we have that $x_1,\ldots,x_n,x_1$ is a cycle in $R$,
      contradicting acyclicity of $R$.
    \item[transitivity:]
      Follows directly from the definition of $S$ and transitivity of $R$.
      \qedhere
  \end{description}
\end{proof}
Combining all these ingredients, we prove our key result on rankings:

\end{toappendix}
We say that the \emph{ranking function~$f$ induces the ranking~${\pref}$}.
Every ranking satisfies \prop{Mirroring}, but not necessarily \prop{Maximal Cut}.
Though focusing on rankings might seem like a restriction, it is not:
\begin{theoremrep}
\label{prop:mirroring-maxcut-ranking}
  For a BMA~$B$, if $\rel{B}$~satisfies \prop{Mirroring} and \prop{Maximal Cut},
  then $\rel{B}$~is a ranking.
\end{theoremrep}
\begin{proofsketch}
  Let $\rel{B}$ satisfy \prop{Mirroring} and \prop{Maximal Cut}.
  The relation~$\rel{B}$ must be acyclic, as from any cycle we can construct a formula violating \prop{Maximal Cut}.
  By \prop{Mirroring}, the elements incomparable wrt.~$\rel{B}$ form equivalence classes,
  and $\rel{B}$~corresponds to a total order over these equivalence classes.
  Mapping each element to its equivalence class, we derive the ranking function~$f$ that induces~$\rel{B}$.
\end{proofsketch}
\begin{proof}
  Let $B$~be a BMA such that $\rel{B}$~satisfies \prop{Mirroring} and \prop{Maximal Cut}.
  By \cref{lem:bma-maxcut-acyclic}, we have that $\rel{B}$~is acyclic.
  Thus, by \cref{lem:irreflexive-mirroring-equivalence-incomp}, we know that the incomparability relation~${\incomp[\rel{B}]}$ is an equivalence relation.
  As established by \cref{lem:eqclass-order-defined,lem:eqclass-order-strictorder},
  the corresponding relation~$S$ on the equivalence classes of~${\incomp[\rel{B}]}$ is a total strict order.
  Hence, the function~$f$ with $f(\tau)=[\tau]$ is a ranking function over the totally-ordered set $(\{\,[\tau]\mid\tau\in\Sigma^\omega\,\},S)$
  and induces the ranking~$\rel{B}$.
\end{proof}

\begin{toappendix}
  Conversely, we show the following observation made in \cref{sec:constructions}:
  \begin{observation}
    Every ranking satisfies \prop{Mirroring}.
  \end{observation}
  \begin{proof}
    Let $f:\Sigma^\omega\to Y$~be the ranking function inducing a ranking~${\pref}$ over the totally-ordered set $(Y,{<_Y})$.
    Then, if $\tau_1 \incomp[\pref] \tau_2$ and $\tau_1 \pref \tau_3$,
    we have $f(\tau_2) = f(\tau_1) <_Y f(\tau_3)$,
    and thus $\tau_2 \pref \tau_3$ as required by \prop{Mirroring}.
  \end{proof}
\end{toappendix}
This result states that, fundamentally,
due to the expressiveness of LTL,
\emph{the only BMA relations guaranteeing full rationality} are rankings:
traces are partitioned into subsets of equi-preferable traces, and these subsets must be \emph{totally} ordered.

While we here state the theorem specifically for LTL and for BMA relations over traces,
the underlying insight is broader.
\Citet{jandson:thesis} shows a representation result between fully-rational AGM contraction operators and \emph{Blade Contraction Functions} induced by relations over complete consistent theories (CCTs), for all boolean Tarskian logics.
In this general setting, \cref{prop:mirroring-maxcut-ranking} carries over to a wide class of expressive logics, the \emph{compendious logics}~\citep{kr25:effectiveAGMContraction}.
We derive the following representation result:
\begin{toappendix}
  In order to prove the general representation result of~\cref{thm:repr-result} over arbitrary \emph{compendious logics},
  we repeat some necessary background information from \citet{jandson:thesis} and \citet{kr25:effectiveAGMContraction}.

  \Citet{jandson:thesis} presents a general formalism for AGM contraction over boolean, Tarskian logics.
  The formalism is based on complete consistent theories (CCTs).
  Epistemic preferences are expressed by a preference relation~$\pref$ over CCTs that satisfies \prop{Mirroring}
  and the following property:
  \begin{description}
    \item[{\propdef[Maximal Cut\textsubscript{CCT}]{MaxCutCCT}}] For every formula~$\varphi$,
      the set~$\compl{\varphi}$, consisting of all CCTs that do not contain~$\varphi$, must have at least one maximal element wrt.~$\pref$.
  \end{description}
  \Citet{jandson:thesis} then shows a representation result between all fully rational AGM contraction operators
  and \emph{Blade Contraction Functions},
  i.e., functions~$\dotmin[\pref]$ of the form
  \[
    \kb \dotmin[\pref] \varphi = \begin{cases}
      \kb \cap \bigcap \max_{\pref} \compl{\varphi} & \textbf{if } \varphi\notin\Cn{\emptyset} \text{ and } \varphi\in\kb\\
      \kb & \textbf{otherwise}
    \end{cases}
  \]
  for some preference relation~${\pref}$ satisfying \prop{Mirroring} and \prop{MaxCutCCT}.

  Our formalism is based on the work of \citet{kr25:effectiveAGMContraction},
  who instantiate the work of \citet{jandson:thesis} for LTL and show an isomorphism between CCTs and ultimately periodic traces.
  \Citet{kr25:effectiveAGMContraction} also describe a general class of expressive (boolean, Tarskian) logics, the \emph{compendious logics},
  of which LTL is an example.
  We refer to~\citep{kr25:effectiveAGMContraction} for the details of the definition.
  Let us merely note that, in extension of what we already stated for LTL above,
  it is shown there that in every such compendious logic,
  there exists an \emph{identifying formula}~$\id{\kb}\in \kb$ for every CCT~$\kb$, such that $\kb$~is the only CCT to contain~$\id{\kb}$.
  We show the following analogue of \cref{lem:bma-maxcut-acyclic}:
  \begin{lemma}
    \label{lem:cct-maxcut-acyclic}
    Consider a compendious logic, and let ${\pref}$~be a binary relation over CCTs of this logic,
    such that ${\pref}$~satisfies \prop{MaxCutCCT}.
    Then ${\pref}$~is acyclic.
  \end{lemma}
  \begin{proof}
    Suppose that there exists a cycle~$\kb_1 \pref \kb_2 \pref \ldots \pref \kb_n \pref \kb_1$.
    Consider the formula $\varphi :\equiv \lnot (\id{\kb_1} \lor \ldots \lor \id{\kb_n})$.
    The CCTs that do not contain~$\varphi$ are exactly $\compl{\varphi} = \{\kb_1,\ldots,\kb_n\}$.
    As these traces form a cycle in~${\pref}$, we have that $\max_{\pref} \compl{\varphi} = \emptyset$.
    This contradicts \prop{MaxCutCCT}.
  \end{proof}
\end{toappendix}
\begin{theoremrep}
  \label{thm:repr-result}
  Every fully rational AGM contraction operator over a compendious logic \citep{kr25:effectiveAGMContraction}
  is a Blade Contraction Function~\citep{jandson:thesis} induced by a \emph{ranking} over CCTs that satisfies \emph{maximal cut}~\cite{jandson:thesis},
  and vice versa.
\end{theoremrep}
\begin{proof}
  We appeal to \citet{jandson:thesis} for the representation result between fully rational AGM contraction operators and Blade Contraction Functions.
  Hence, it remains only to show that in compendious logics, the preference relations~${\pref}$ over CCTs satisfying \prop{Mirroring} and \prop{MaxCutCCT} are rankings (given that we already observed that rankings satisfy \prop{Mirroring}).
  We proceed analogously to the proof of~\cref{prop:mirroring-maxcut-ranking}.

  Let ${\pref}$~be a relation over CCTs (of a compendious logic) that satisfies \prop{Mirroring} and \prop{MaxCutCCT}.
  By \cref{lem:cct-maxcut-acyclic}, we have that ${\pref}$~is acyclic.
  Thus, by \cref{lem:irreflexive-mirroring-equivalence-incomp}, we know that the incomparability relation~${\incomp[{\pref}]}$ is an equivalence relation.
  As established by \cref{lem:eqclass-order-defined,lem:eqclass-order-strictorder},
  the corresponding relation~$S$ on the equivalence classes of~${\incomp[{\pref}]}$ is a total strict order.
  Hence, the function~$f$ with $f(\kb)=[\kb]$ is a ranking function over the totally-ordered set $(\{\,[\kb]\mid \kb \in \CCT\,\},S)$
  and induces the ranking~${\pref}$.
\end{proof}
\begin{toappendix}
  \subsection{Proofs for BMA Constructions}
  The remainder of this appendix section is dedicated to the different constructions of well-founded BMA rankings.
\end{toappendix}
Consequently, we confine our study of BMA relations to rankings.
We call such rankings \emph{BMA rankings}.
To ensure \prop{Maximal Cut},
we focus on rankings over well-founded sets.
\begin{observation}
  If ${\pref}$ is a ranking over~$(Y,{<_Y})$ and ${<_Y}$ is well-founded,
  then ${\pref}$ is well-founded.
\end{observation}
In the remainder of this section, we provide three constructions of BMA relations that represent well-founded rankings,
and one construction to \emph{hierarchically compose} multiple BMA rankings while preserving well-foundedness.
Together, these constructions induce a large class of BMA rankings and corresponding fully-rational contractions.
\paragraph{(1) Lists of Ranks.}
The most straightforward way to ensure \prop{Well-Founded} is to consider finite rankings.
Let therefore $M_1,\ldots, M_n$ be a partition of the set~$\Sigma^\omega$
(i.e., $\bigcup_{i=1}^n M_i=\Sigma^\omega$ and $M_i\cap M_j=\emptyset$ for $i\neq j$).
Assume that the traces in $M_1$ are considered less plausible than those in $M_2$,
those in $M_1 \cup M_2$ are less plausible than those in $M_3$, etc.
The corresponding ranking function~$f$ maps traces to the well-ordered set $(\{1,\ldots,n\}, {<_\N})$,
and satisfies $f(\tau) = i$ if and only if $\tau \in M_i$.
If each rank~$M_i$
is described as a B\"uchi automaton,
the induced ranking can be represented as a BMA.
\begin{toappendix}
  For several of these constructions, we use the following lemma:
  \begin{lemma}
    \label{lem:buchi-product-bma}
    Let $A_1,A_2$ be B\"uchi automata.
    Then there exists a BMA~$B$ with $\rel{B} = \lang{A_1} \times \lang{A_2}$.
  \end{lemma}
  \begin{proof}
    For the cartesian product of B\"uchi automata $A_i = (Q_i, \Sigma, \Delta_i, Q_0^i, R_i)$ and $A_j = (Q_j, \Sigma, \Delta_j, Q_0^j, R_j)$,
    we construct a BMA $B=(Q_B, \Sigma_{\textnormal{BM}}, \Delta_B, Q_0^B, R_B)$
    with
    $Q_B=Q_i\times Q_j \times \{0,1\}$,
    $Q_0^B = Q_0^i \times Q_0^j \times \{0\}$,
    $R_B = \{\,\langle q_0, q_1, k\rangle \in Q_B \mid q_k \in R_k \,\}$,
    and $(\langle q_0, q_1, k \rangle, a/b, \langle q_0', q_1', k' \rangle) \in \Delta_B$
    if and only if $(q_0,a,q_0')\in \Delta_i$, $(q_1,b,q_1')\in\Delta_j$,
    and $k'=(k+1)\bmod 2$ if $q_k\in R_k$ or $k'=k$ otherwise.
    Similar to the canonical construction for the intersection of B\"uchi automata,
    the flag~$k$ is used to ensure that both $A_i$~and~$A_j$ pass through one of their respective recurrence states infinitely often.
  \end{proof}
\end{toappendix}
\begin{propositionrep}
  \label{prop:list-rank-bma}
  Let $A_1,\ldots, A_n$ be B\"uchi automata with disjoint languages,
  such that $\bigcup_{i=1}^n\lang{A_i} = \Sigma^\omega$.
  The ranking function~$f$ with $f(\tau)=i$ if and only if~$\tau\in\lang{A_i}$
  induces a well-founded BMA ranking.
\end{propositionrep}
\begin{proof}
  Let ${\pref}$~be the ranking induced by~$f$.
  We observe that this ranking can be written as a finite union of cartesian products:
  \[
    {\pref} = \bigcup_{i=1}^n \bigcup_{j=i+1}^{n} \lang{A_i} \times \lang{A_j}
  \]
  As BMA are closed under union (\cref{lem:bma-closed-ops}),
  and the cartesian product between the languages of two B\"uchi automata can be recognized by a BMA (\cref{lem:buchi-product-bma}),
  we conclude that ${\pref}$~is a BMA ranking.
\end{proof}
It is natural to define the preference relation in terms of the formulae satisfied by the traces.
In particular, we can choose automata~$A_i$ corresponding to LTL formulae.
\begin{example}
\label{ex:next-days-rain-ranking}
Let the atomic proposition~$r$ denote that \quot{it rains}.
We may consider it more plausible that \quot{it will rain tomorrow} (in LTL, $\varphi :\equiv \Next r$) than that \quot{it will not rain tomorrow} $(\lnot\varphi$).
This preference corresponds to a BMA relation, as we can see by setting $A_1 = \ltlbuchi{\lnot\varphi}$, $A_2 = \ltlbuchi{\varphi}$ and $n=2$ in \cref{prop:list-rank-bma}.

LTL also allows us to express more fine-grained temporal preferences.
We can express the attitude that most plausibly, \quot{it will rain today},
somewhat less plausibly \quot{it will rain tomorrow (but not today)},
and even less plausibly \quot{it will eventually rain (but not today or tomorrow)}.
The alternative \quot{it will never rain again} is most implausible.
To express this, we set~$n$ to~$4$ and $A_i:=\ltlbuchi{\varphi_i}$ (for $i=1,\ldots,n$), with
\[
  \varphi_1 \equiv \Globally\lnot r
  \qquad\quad
  \varphi_2 \equiv \Finally r \land (\lnot\Next r) \land \lnot r
  \qquad\quad
  \varphi_3 \equiv (\Next r)\land\lnot r
  \qquad\quad
  \varphi_4 \equiv r
\]
The BMA recognizing the corresponding ranking is shown in~\cref{fig:bma-example} and discussed in \cref{ex:bma-preferences}.
\end{example}
While often useful,
this approach allows us to distinguish only finitely many ranks of equi-preferable traces.
Consider again the example above,
but with the preference that \quot{it is plausible that it rains \emph{soon}}.
Or more precisely, \quot{the sooner it rains (according to some trace), the more plausible}.
This requires us to distinguish infinitely many ranks (\quot{it first rains in $n$ days}, for each~$n\in\N$).
As we show next, BMA are expressive enough to represent such preferences.
\paragraph{(2) The sooner, the better.}
A rich class of infinite rankings can be described by ordering traces according to the earliest time point from which on some plausible LTL formula~$\varphi$ holds.
I.e., we prefer traces satisfying $\varphi$
over traces satisfying only $\Next \varphi$,
over traces satisfying only $\Next[2] \varphi$, etc.
Given a trace~$\tau$,
we write $\earliestFuture{\varphi}{\tau}=i$ for the minimal~$i\in\N$
such that $\tau  \models \Next[i] \varphi$,
and $\earliestFuture{\varphi}{\tau} = \infty$ if no such~$i$ exists.
\begin{propositionrep}
  \label{prop:earliest-future-bma}
  The ranking function ${f:\Sigma^\omega \to \Z_{\leq0}\cup\{-\infty\}}, \tau \mapsto -\earliestFuture{\varphi}{\tau}$ induces a well-founded BMA ranking.
\end{propositionrep}
\begin{proof}
  It is straightforward to see that the ranking is well-founded.
  To construct a BMA recognizing the ranking,
  let $B_1=(Q_1,\Sigma_{\textnormal{BM}}, \Delta_1, Q_0^1, R_1)$ be a BMA recognizing $\rel{B_1} = \lang{\ltlbuchi{\lnot\varphi}}\times \lang{\ltlbuchi{\lnot\varphi}}$,
  and let $B_2 = (Q_2, \Sigma_{\textnormal{BM}}, \Delta_2, q_0^2, R_2)$ be a BMA recognizing $\rel{B_2} = \lang{\ltlbuchi{\lnot\varphi}}\times \lang{\ltlbuchi{\varphi}}$.
  Such BMA exist by~\cref{lem:buchi-product-bma}.
  We assume wlog.\ that $Q_0^1$ and $Q_0^2$ are disjoint.

  For convenience, we construct an \emph{alternating} B\"uchi-Mealy automaton for the ranking induced by~$f$.
  That is, each transition is represented by a tuple $(q, x/y, Q') \in \Delta \subseteq Q \times \Sigma_{\textnormal{BM}} \times \powerset{Q}$,
  with the meaning that from \emph{every} state in~$Q'$, there must be a run visiting a recurrence state infinitely often.
  It is well known that alternating B\"uchi automata (and hence, BMA) have the same expressive power as normal B\"uchi automata (resp.~BMA).

  Specifically, we set $B := (Q_B, \Sigma_{\textnormal{BM}}, \Delta_B, Q_0^B, R_B)$,
  with states $Q_B = \{q_\init\} \cup Q_1 \cup Q_2$,
  a single initial state~$Q_0^B = \{q_\init\}$,
  recurrence states $R_B = R_1 \cup R_2$,
  and the transition relation
  \begin{align*}
    \Delta_B =\null & \big\{\, (q_\init, x/y, \{q_\init, q'\}) \mid x,y\in\Sigma \text{ and } (q_0, x/y, q')\in\Delta_1 \text{ for some } q_0\in Q_0^1 \,\big\}\\
                    & \null \cup \big\{\, (q_\init, x/y, \{q'\}) \mid x,y\in\Sigma \text{ and } (q_0, x/y, q')\in \Delta_2 \text{ for some } q_0\in Q_0^2 \,\big\} \\
                    & \null \cup \Delta_1 \cup \Delta_2
  \end{align*}
  While the transitions from the first set are taken,
  the BMA~$B$ makes sure that neither the first nor the second trace in a pair~$(\tau_1,\tau_2)\in\rel{B}$ satisfy~$\varphi$.
  At a nondeterministically chosen point in time~$i$, $B$~may instead choose a transition from the second set,
  ensuring that at time point~$i$, the trace~$\tau_1$ still does not satisfy~$\varphi$ (and hence, $\earliestFuture{\varphi}{\tau_1}>i$)
  whereas the trace~$\tau_2$ does satisfy~$\varphi$ (i.e., $\earliestFuture{\varphi}{\tau_2} = i$).
  The existence of such a point in time~$i$ corresponds exactly to the condition~$f(\tau_1) < f(\tau_2)$.
  The BMA~$B$ enforces that a transition of the second set is eventually taken (and thus such an~$i$ exists),
  as otherwise there is a run~$q_\init^\omega$ which never visits a recurrence state.
\end{proof}
The induced BMA ranking is denoted as~$\pref_\varphi$.
We invert the sign, as traces where the desired property holds at a \emph{smaller} time index are considered \emph{more} preferable.
Well-foundedness arises from the fact that there does not exist an infinite descending chain of trace indices~$i\in\N$.

~\vspace*{-2em}
\begin{wrapfigure}[5]{r}{0.17\textwidth}
\centering
\vspace*{-1em}
\begin{tikzpicture}[thick]
  \node[draw,circle] (q0) {};
  \node[draw,double,circle,right=10mm of q0] (q1) {};
  \draw[<-] (q0) -- ++(left:7mm);
  \draw[->] (q0) edge[loop above] node[auto,font=\scriptsize,align=center]{$\emptyset/\emptyset$,\ \ $\emptyset/\{s\}$,\\$\{s\}/\emptyset$, \\ $\{s\}/\{s\}$} ();
  \draw[->] (q0) -- node[above,font=\scriptsize] {$\{s\}/\emptyset$} (q1);
  \draw[->] (q1) edge[loop above] node[auto,font=\scriptsize,align=center]{$\emptyset/\emptyset$,\\$\{s\}/\{s\}$} ();
\end{tikzpicture}
\captionsetup{hangindent=0pt,indention=0pt,labelfont=bf,format=plain}
\caption{BMA for the ranking $\pref_{\Globally\lnot s}$.}
\label{fig:bma-first-noswnow}
\end{wrapfigure}
\begin{example}
  The BMA ranking~$\pref_\varphi$, for the LTL formula~$\varphi\equiv r$, expresses the preference \quot{it is plausible that it rains soon} discussed above.
  Our result also applies to more complex temporal statements.
  For instance, \quot{it is plausible that soon, it will never snow again} is expressed by the BMA ranking~$\pref_\psi$, for the LTL formula~$\psi \equiv \Globally\lnot s$,
  where the proposition~$s$ indicates that \quot{it snows} (see \cref{fig:bma-first-noswnow}).
\end{example}
\paragraph{(3) Rankings from Distance Measures.}
A well-known approach to construct epistemic preference relations in the classical setting
is based on defining a \emph{distance} between models (propositional valuations).
The most preferred models are those with the least distance to the models of the knowledge base~$\kb$ from which a formula is contracted.
I.e., models are \emph{ranked} by their distance from~$\kb$, and the closest models are most preferred.
The most prominent representative of this approach is the distance defined by~\citet{Dalal88,DalalTc}.
Many other distance measures fail to ensure rationality~\citep{RibeiroT21} in the classical setting.

We define a well-founded BMA ranking
based on a notion of distance between the traces and the given knowledge base~$\kb$.
The construction is parametric in a distance function~$d : \Sigma\times\Sigma \to \R_{\geq 0} \cup\{+\infty\}$ on propositional valuations in~$\Sigma=\powerset{\AP}$.
It may be useful to think of~$d$ as a distance in the mathematical sense that it satisfies the \emph{metric axioms}%
(symmetry, triangle inequality, and $d(x,y)=0$ iff $x=y$).
For instance, the distance of~\citet{Dalal88,DalalTc} (also known as \emph{Hamming distance}) satisfies these axioms.
Yet, neither rationality nor the constructive realization with BMA depend on such assumptions,
thus allowing for the use of many broader notions of distance that fail rationality in the classical usage.

Observe that the distance~$d(x,y)$ can only take finitely many values, as $\Sigma\times\Sigma$ is finite;
let the range of~$d$ be $\{r_0,r_1,\ldots,r_n\}$
with $r_0 < r_1 < \ldots < r_n$.
The function~$d$ gives rise to a distance on traces:
\begin{equation*}
  d_\infty(A_0A_1\ldots, B_0B_1\ldots) := \textstyle \max_{i\in\N} d(A_i,B_i)
\end{equation*}
\begin{figure}[t]
\centering
\begin{subfigure}{0.5\textwidth}
\centering
\tikzstyle{pv}=[font=\scriptsize,inner sep=2,fill=white]
\begin{tikzpicture}[thick,yscale=0.5]
  \foreach \i in {0,1,...,6}{
    \draw[thin,dotted,gray] (\i,-0.5) -- (\i,5.5);
  }
  \draw[->,very thick,gray] (-0.5,-0.5) -- node[pos=0.9,below,font=\scriptsize]{Time} (6.5,-0.5);

  \node[pv,text=orange] (20) at (0,3.0) {$\{p\}$};
  \node[pv,text=orange] (21) at (1,4.5) {$\{p,r\}$};
  \node[pv,text=orange] (22) at (2,3.0) {$\{p,r\}$};
  \node[pv,text=orange] (23) at (3,1.5) {$\{p,q\}$};
  \node[pv,text=orange] (24) at (4,4.5) {$\{p,s\}$};
  \node[pv,text=orange] (25) at (5,3.0) {$\{p,s\}$};
  \node[pv,text=orange] (26) at (6,3.0) {$\{p,s\}$};
  \node[orange,right=-1mm of 26] {$\mathbf{\cdots}$};
  \node[orange,left=-1.5mm of 20,font=\scriptsize] {$\tau_1=$};

  \draw[orange] (20) -- (21)
        (21) -- (22)
        (22) -- (23)
        (23) -- (24)
        (24) -- (25)
        (25) -- (26);

  \node[pv,text=blue] (10) at (0,0) {$\{q\}$};
  \node[pv,text=blue] (11) at (1,0) {$\{q\}$};
  \node[pv,text=blue] (12) at (2,0) {$\emptyset$};
  \node[pv,text=blue] (13) at (3,0) {$\{q\}$};
  \node[pv,text=blue] (14) at (4,0) {$\{q\}$};
  \node[pv,text=blue] (15) at (5,0) {$\emptyset$};
  \node[pv,text=blue] (16) at (6,0) {$\emptyset$};
  \node[blue,right=-1mm of 16] {$\mathbf{\cdots}$};
  \node[blue,left=-1.5mm of 10,font=\scriptsize] {$\tau_2=$};

  \draw[blue] (10) -- (11)
        (11) -- (12)
        (12) -- (13)
        (13) -- (14)
        (14) -- (15)
        (15) -- (16);

  \draw[red,very thick,<->] (14) -- node[right,font=\scriptsize,pos=0.4,xshift=-0.7mm]{$d_\infty(\tau_1,\tau_2)=3$} (24);
\end{tikzpicture}
\caption{Trace distance~$d_\infty$, where $d$ is the Dalal distance.}
\label{fig:trace-distance}
\end{subfigure}
\hfill
\begin{subfigure}{0.4\textwidth}
\centering
\begin{tikzpicture}[scale=0.8]
  \draw[orange,fill=orange!40!white,very thick] (0,0) ellipse (2 and 1.25);
  \node[orange,font=\scriptsize,xshift=6mm] at (1.5,1.25) {$\{\,\sigma\mid \sigma \models \kb \,\}$};
  \node[inner sep=1] (K')   at (-2,1.5) {$\tau$};
  \node[inner sep=1] (tau)  at (-1.15,0.65) {$\sigma$};
  \node[inner sep=1,xshift=-2mm] (tau') at (1.5,-0.5) {$\sigma'$};
  \draw[red,very thick,<->] (K') -- node[left,rotate=40,font=\scriptsize]{$-f_d^{\kb}(\tau)$} (tau);
  \node[inner sep=1] (K'')  at (3.5,-1) {$\tau'$};
  \draw[red,very thick,<->] (K''.west) -- node[right,rotate=40,pos=0.5,font=\scriptsize]{$-f_d^{\kb}(\tau')$} (tau'.east);
\end{tikzpicture}
\caption{Ranking function $f_d^{\kb}$ derived from the distance~$d_\infty$ to the closest~$\sigma$ resp.\ $\sigma'$. In the induced ranking, $\tau$ is preferred over $\tau'$.}
\label{fig:distance-based-rank}
\end{subfigure}
\caption{Construction of a ranking from a notion of distance~$d$ on propositional valuations.}
\end{figure}
\Cref{fig:trace-distance} illustrates this distance measure on traces.%
If $d$~is a metric, the distance~$d_\infty$ is the corresponding \emph{Chebyshev metric}.
For a consistent theory~$\kb$,
we define the rank~$f_d^\kb(\tau)$ of a trace~$\tau$
via the distance between $\tau$~and the closest trace~$\sigma$ such that $\sigma\models\kb$ holds:
\[
  f_d^\kb(\tau) := -\min \{\, d_\infty(\tau, \sigma) \mid \sigma \models \kb \,\}
\]
We invert the sign such that traces~$\tau$ close to~$\kb$ are preferred.
The minima and maxima in these equations are well-defined, as they operate on nonempty subsets of the finite set $\{r_0,\ldots,r_n\}$.
\Cref{fig:distance-based-rank} illustrates the ranking function.%
And indeed, this ranking (denoted~$\pref_d^\kb$) can be represented as BMA:
\begin{toappendix}
  \begin{lemma}
    \label{lem:bma-exists-projection}
    Let $B$~be a BMA. There exists a B\"uchi automaton~$A$ such that
    \[
      \lang{A} =\{\,\tau\in\Sigma^\omega \mid \exists \tau'\in\Sigma^\omega\,.\, (\tau,\tau')\in\rel{B}\,\}\ .
    \]
  \end{lemma}
  \begin{proof}
    We obtain~$A$ from~$B$ by replacing every transition labeled with some pair $x/y$ by a transition labeled only with~$x$.
  \end{proof}
\end{toappendix}
\begin{propositionrep}
  \label{prop:measure-bma}
  If $\kb\in\excBuchi$~is consistent,
  the ranking function~$f_d^{\kb}$ induces a well-founded BMA ranking.
\end{propositionrep}
\begin{proof}
  As $f_d^\kb$~ranges over the totally-ordered finite set~$\{-r_n,\ldots, -r_0\}$,
  the ranking is well-founded.

  To see that the ranking is recognized by a BMA,
  we construct B\"uchi automata $A_0,\ldots,A_n$ such that $\lang{A_i} = \{\,\tau\in\Sigma^\omega \mid f_d^\kb(\tau) = r_i \,\}$.
  It then suffices to apply \cref{prop:list-rank-bma} to the sequence~$A_n,\ldots,A_0$.

  Towards a construction for the~$A_i$,
  let $\kb=\support{A}$ for a B\"uchi automaton~$A$,
  and observe:
  \[
    f_d^{\kb}(\tau) = r_i \iff (\exists \tau' \in \lang{A}\,.\, d_\infty(\tau, \tau') = r_i) \land \bigwedge_{j<i}\lnot\exists \tau''\in\lang{A}\,.\, d_\infty(\tau,\tau'') = r_j
  \]
  We first construct a BMA~$B_i$ such that $\rel{B_i}=\{\,(\tau,\tau') \mid d_\infty(\tau,\tau') = r_i\,\}$.
  Specifically, we set $B_i = (\{r_0,\ldots,r_n\},\Sigma_\textrm{BM}, \Delta_d, \{r_0\}, \{r_i\})$,
  with
  $
    \Delta_d = \{\, (r, x/y, r') \mid r' = \max(r, d(x,y)) \,\}
  $.
  This BMA keeps track of the maximal distance between the already-read prefixes of the two traces~$\tau,\tau'$.
  The maximal distance can only increase (never decrease) while reading a word.
  Hence, an accepting run must eventually reach and then forever remain in the state~$r_i$.
  I.e., there is an index where the traces have distance~$r_i$, and no index where they have greater distance: $d_\infty(\tau,\tau')=r_i$.
  Conversely, if $d_\infty(\tau,\tau')=r_i$, then the state $r_i$ is reached after finitely many steps, and no greater distance is observed,
  so the run for the pair $(\tau,\tau')$ is accepting.

  We apply \cref{lem:buchi-product-bma} for the existence of a B\"uchi automaton recognizing~$\Sigma^\omega \times \lang{A}$.
  By applying \cref{lem:bma-exists-projection} to the intersection of this automaton with~$B_i$,
  we derive a B\"uchi automaton~$A'_i$
  such that $\lang{A'_i} = \{\,\tau \mid \exists \tau'\in\lang{A}\,.\, d_\infty(\tau,\tau') = r_i\,\}$.
  Noting that
  $
    \lang{A_i} = \lang{A'_i} \cap \bigcap_{j<i} \Sigma^\omega \setminus \lang{A'_j}\ ,
  $
  we derive $A_i$ using standard constructions for the intersection and complementation of B\"uchi automata.
\end{proof}
\paragraph{(4) Hierarchical Preferences.}
A powerful tool when describing (ranked) preferences is the ability to express \emph{hierarchical} preferences.
For instance, when discussing the weather,
we might prefer (resp.\ consider as more plausible) traces in which \quot{soon, it will never snow again},
but among two traces that agree on the last snow day,
we might prefer the trace that is closer to our belief base~$\kb$.

Given ranking functions~$f_i : \Sigma^\omega \to Y_i$ over totally-ordered sets~$(Y_i,{<_i})$, for $i\in\{1,2\}$,
a hierarchical preference can be expressed as the ranking function
$
  f_{12} : \Sigma^\omega \to Y_1\times Y_2$ with $f(\tau) = (f_1(\tau), f_2(\tau))
$,
where $Y_1\times Y_2$ is totally ordered by the lexicographical combination of~${<_1},{<_2}$,
i.e., $(y_1,y_2) <_{12} (y_1',y_2')$ if $y_1 <_1 y_1'$, or $y_1=y_1'$ and $y_2 < y_2'$.
BMA rankings are closed under this hierarchical composition.
\begin{propositionrep}
  Let ${\pref_1},{\pref_2}$ be BMA rankings induced by ranking functions~$f_1$ resp.~$f_2$.
  The ranking function~$f_{12}$ induces a BMA ranking~${\pref_{12}}$.
  If ${\pref_1}$~and~${\pref_2}$ are well-founded, so is the ranking~${\pref_{12}}¨$.
\end{propositionrep}
\begin{proof}
  We begin with the following equivalence:
  \begin{align*}
    \tau_1 \pref_{12} \tau_2
    &\null\iff
      f_{12}(\tau_1) = (f_1(\tau_1), f_2(\tau_1)) <_{12} (f_1(\tau_2), f_2(\tau_2)) = f_{12}(\tau_2)\\
    &\null\iff
      f_1(\tau_1) <_1 f_1(\tau_2) \lor \big( f_1(\tau_1) = f_1(\tau_2) \land f_2(\tau_1) <_2 f_2(\tau_2) \big)\\
    &\null\iff
      \tau_1 \pref_1 \tau_2 \lor (\tau_1 \incomp[\pref_1] \tau_2 \land \tau_1 \pref_2 \tau_2)
  \end{align*}
  where ${\incomp[\pref_1]}$~denotes the incomparability relation corresponding to~${\pref_1}$,
  ie., $y \incomp[\pref_1] y'$ holds if and only if $\lnot(y \pref_1 y')$ and $\lnot(y' \pref_1 y)$ both hold.
  Thus, we have ${\pref_{12}} = {\pref_1} \cup ({\incomp[\pref_1]} \cap {\pref_2})$,
  with ${\incomp[\pref_1]} = (\Sigma_\textrm{BM}^\omega \setminus {\pref_1}) \cap (\Sigma_\textrm{BM}^\omega \setminus {\pref_1}^{-1})$.
  Using~\cref{lem:bma-closed-ops}, we conclude that ${\pref_{12}}$~is a BMA relation.

  Finally, to see that ${\pref_{12}}$~is well-founded if ${\pref_1}$~and~${\pref_2}$ are well-founded,
  assume (for purposes of contradiction) that there existed an infinite ascending sequence $\tau_1 \pref_{12} \tau_2 \pref_{12} \cdots$.
  As ${\pref_1}$~is by assumption well-founded,
  from some point~$n$ on all the traces~$\tau_i$ (with $i \geq n$) must be in the same rank wrt.~${\pref_1}$.
  But then it follows that $\tau_n \pref_2 \tau_{n+1} \pref_2 \ldots$ is an infinite ascending sequence wrt.~${\pref_2}$,
  which contradicts our assumption that ${\pref_2}$~is well-founded.
\end{proof}
We denote the composed ranking as~${\pref_1} \triangleright {\pref_2}$.
In combination with the other constructions presented above,
hierarchical preference composition allow us to express complex preference relations.
\begin{example}
  Consider the BMA ranking~$\pref_\psi$, for $\psi = \Globally\lnot s$ (\quot{it is plausible that soon, it will never snow again}),
  and the BMA ranking~$\pref_\textsf{Dalal}^\kb$ induced by the Dalal distance measure.
  By composing the two rankings, we arrive at an epistemic preference relation~${\pref_\psi} \triangleright {\pref_\textsf{Dalal}^\kb}$
  that prefers traces in which the last snow day occurs sooner rather than later;
  and among traces that agree on the last snow day,
  it prefers traces that are closer to~$\kb$ (as defined by the Dalal-Chebyshev distance).
  Through this composition,
  each of the infinitely many ranks of~${\pref_\psi}$
  is sub-divided into a finite number of sub-ranks based on the Dalal distance.
\end{example}

\section{Conclusion}
\label{sec:conclusion}

Epistemic preference relations form the backbone of rational belief change.
We have studied the particular case of effective AGM contraction in an expressive, non-finitary logic, namely linear temporal logic (LTL).
In this setting, preference relations are represented via so-called \emph{B\"uchi-Mealy automata}
and must satisfy two core properties to ensure rationality of the corresponding contraction.
We have shown that while one of the properties, \prop{Mirroring}, is decidable,
the key property ensuring success of the contraction, \prop{Maximal Cut}, is undecidable.
This has far-reaching consequences, as it rules out the possibility of a representation result for effective LTL contraction.
Consequently, effective LTL contraction requires custom design of suitable preference relations for particular applications.

We have proposed a number of constructions for preference relations
that guarantee the required properties \emph{by design};
based on the observation that such relations must necessarily be \emph{rankings},
potentially with an infinite number of ranks.
One of our constructions generalizes an idea from classical belief change,
namely the construction of preferences based on the distance to the models of the current epistemic state.
	Constructing preference relations for classical propositional logics is also a challenge. 
Although some approaches were proposed~\citep{Dalal88,KatsunoM92}, very few achieve full rationality, being Dalal and some \textit{culpability measures}~\citep{RibeiroT21} the few known ones. %
Dalal can only be used on consistent beliefs, %
while \textit{culpability measures} work only for inconsistent beliefs. %
Our approaches easily apply for propositional logics, and provide novel and richer alternatives to obtain preference relations.

The attentive reader may have observed that we have not shown the undecidability of \prop{Maximal Cut} \emph{assuming that \prop{Mirroring} holds},
or equivalently, the undecidability of the conjunction of the two properties.
We leave the full resolution of this question to future work.
However, let us remark that there are strong indications that undecidability indeed persists under such a restriction.
In particular, our proof of undecidability for \prop{Maximal Cut} relies on a reduction from a variation of the uniform halting problem,
i.e., the question whether a given Turing machine \emph{always terminates}.
In the study of \emph{program verification},
a link has long been made between termination and the existence of a (well-founded) \emph{ranking function} (sometimes also called \emph{variant} or \emph{termination measure}).
We suspect that closing the loop to our previous observation that relations satisfying both \prop{Mirroring} and \prop{Maximal Cut} must indeed be rankings
would allow for such a proof (though the technical details remain challenging).

\section*{Declaration on Generative AI}
  The author(s) have not employed any Generative AI tools.

\bibliography{references}

\end{document}